\documentclass[11pt]{article}
\usepackage[margin=1in]{geometry}
\usepackage[english]{babel}
\usepackage{xr}
\usepackage{booktabs}
\usepackage{setspace}
\usepackage{cancel}
\usepackage[normalem]{ulem}

\usepackage{tikz}
\usetikzlibrary{arrows,automata,positioning}

\usepackage{mathrsfs,hyperref, algorithm}%, algorithmic}
\usepackage[]{amsmath}
\usepackage{amsthm}
\usepackage{fix-cm}
\usepackage[]{amssymb}
\usepackage[]{latexsym}
\usepackage[right]{eurosym}
\usepackage[T1]{fontenc}
\usepackage[]{graphicx}
\usepackage[]{epsfig}
\usepackage{fancyhdr}
\usepackage{pstricks}
\usepackage{multirow}
\usepackage{natbib}
\usepackage{subcaption}

\allowdisplaybreaks

\makeatletter

\newcommand{\bbr}{\mathbb{R}}
\newcommand{\E}{\mathbb{E}}

\newcommand{\bbp}{\mathbb{P}}

\newcommand{\fcal}{\mathcal{F}}

\newcommand{\dcal}{\mathcal{D}}

\newcommand{\scal}{\mathcal{S}}
\newcommand{\ncal}{\mathcal{N}}
\newcommand{\lcal}{\mathcal{L}}

\newcommand{\ccal}{\mathcal{C}}

\newcommand{\xcal}{\mathcal{X}}

\newcommand{\LL}{\mathbb{L}}

\newcommand{\N}{\mathbb{N}}

\newcounter{modcount}
\newcommand{\modulo}[2]{%
\setcounter{modcount}{#1}\relax
\ifnum\value{modcount}<#2\relax
\else\relax
\addtocounter{modcount}{-#2}\relax
\modulo{\value{modcount}}{#2}\relax
\fi}
\newcommand{\tablepictures}[4][c]{\begin{tabular}[#1]{@{}c@{}}#2\vspace{0.5cm}\\(\alph{#4}) #3\end{tabular}}
\newcounter{gridsearch}
\newcommand{\tabpic}[2]{
    \stepcounter{gridsearch}
    \modulo{\thegridsearch}{2}
    \ifnum\value{modcount}=0
        \tablepictures[t]{#1}{#2}{gridsearch}\\[2.0cm]
    \else
        \tablepictures[t]{#1}{#2}{gridsearch}&~&
    \fi
}

\newcommand*\circled[1]{\tikz[baseline=(char.base)]{
            \node[shape=circle,draw,inner sep=2pt] (char) {#1};}}

\makeatother
\newtheorem{lemma}{Lemma}[section]
\newtheorem{proposition}[lemma]{Proposition}
\newtheorem{theorem}[lemma]{Theorem}
\newtheorem{corollary}[lemma]{Corollary}
\newtheorem{definition}[lemma]{Definition}
\newtheorem{example1}[lemma]{Example}
\newtheorem{rem1}[lemma]{Remark}

\newtheorem{alg1}[lemma]{Algorithm}
\newtheorem{me1}[lemma]{Mechanism}

\newenvironment{remark}{\begin{rem1}\rm}{\end{rem1}}
\newenvironment{example}{\begin{example1}\rm}{\end{example1}}

\usepackage{color}
\newcommand{\T}{\top}
\newcommand{\diag}{\operatorname{diag}}

\DeclareMathOperator*{\FIX}{FIX}

\DeclareMathOperator*{\argmin}{arg\,min}

\DeclareMathOperator*{\essinf}{ess\,inf}

\newcommand\ind[1]{\mathbb{I}_{\{#1\}}}

\begin{document}

\title{Optimal Network Compression}
\author{Hamed Amini \thanks{Robinson College of Business, Georgia State University, Atlanta, GA 30303, USA, email: {\tt hamini@gsu.edu}} \and
Zachary Feinstein \thanks{Stevens Institute of Technology, School of Business, Hoboken, NJ 07030, USA. \tt{zfeinste@stevens.edu}}}
\date{\today}
%\date{August 2020}
\maketitle
\abstract{
This paper introduces a formulation of the optimal network compression problem for financial systems.  This general formulation is presented for different levels of network compression or rerouting allowed from the initial interbank network.  We prove that this problem is, generically, NP-hard.  We focus on objective functions generated by systemic risk measures under shocks to the financial network. 
We use this framework to study the (sub)optimality of the maximally compressed network.
We conclude by studying the optimal compression problem for specific networks; this permits us to study, e.g., the so-called robust fragility of certain network topologies more generally as well as the potential benefits and costs of network compression.
In particular, under systematic shocks and heterogeneous financial networks the robust fragility results of~\cite{AOT15} no longer hold generally.

\bigskip

}

\section{Introduction}\label{sec:intro}

%\subsection{Motivation}\label{sec:intro-motivation}

The financial crisis 2007-2009 has highlighted the importance of network structure on the amplification of the initial shock to the level of the global financial system, leading to an economic recession. In response to market dysfunctions, the US congress enacted the largest regulations of the financial market, in the form of the "Dodd-Frank Wall Street Reform and Consumer Protection Act" of 2010, to ensure financial stability and reduce systemic risk. Among the regulations is that the majority of over-the-counter (OTC) derivatives should be centrally cleared so as to reduce counterparty risk and ensure financial stability.
Portfolio compression is another way to modify the financial network structure.\footnote{We interchangeably use the terms `portfolio compression' and `network compression' to mean the process of modifying the financial network structure by reducing the gross positions in the network while keeping net positions unchanged.  Though `portfolio compression' is more common in the literature (see, e.g.,~\cite{d2019compressing,veraart2019does}), we introduce the term `network compression' as it explicitly describes the shrinking of the network size being studied.} Several parties in the network enter into a multi-lateral netting agreement to essentially reduce the gross exposures while keeping the net positions unchanged.
The main provider of such systems is TriOptima~\cite{TriOptima}, who have compressed over $\$2$ quadrillion in gross notional. 

%We use interchangeably the terms portfolio compression and network compression to mean modifying the financial network structure by reducing the gross positions in the network while keeping net positions unchanged; see Section~\ref{sec:optimization}.

For the purposes of this paper we consider a \emph{known} initial finite network of obligations over which we seek an ``optimal'' network compression.  This is in contrast to the random graph structure considered in \cite{ACM16, EGJ14, GK10}.  Though the initial network compression formulation is presented without consideration of the network clearing procedure, we will primarily focus on clearing based on \cite{EN01,RV13} with collateral as presented in, e.g.,~\cite{ghamami2021collateralized,veraart2019does}.  Under the DebtRank \cite{battiston2012debtrank,bardoscia2015debtrank} clearing, a version of the optimal network compression problem as a mixed integer linear program was proposed in \cite{diem2020minimal}.  Other notions of contagion could be added to our clearing problem as well, e.g., portfolio overlap \cite{CFS05,AFM16,feinstein2015illiquid}; such additional avenues of contagion would influence the systemic risk and may impact the optimal compression.  We focus on the collateralized Eisenberg-Noe framework so as to remain in a, relatively, simple setting.

\subsection{Literature Review}\label{sec:intro-lit}

The optimal compression problem is related to studies in many other works in the Eisenberg-Noe clearing framework.  For instance, the compression constraints can be viewed as the feasibility conditions for a network reconstruction problem via the network rerouting problem (see~Example \ref{ex:compression}) in which only the aggregate statistics for each bank need be known; \cite{ACP14,UW04,HalajKok:ECB,GV16,GV19} propose methods to sample either deterministically or stochastically from this feasible region.  \cite{AOT15} considered the optimal rerouting problem of a system of identical banks under i.i.d.\ Bernoulli shocks.  That work found that the completely connected system has a ``robust fragility'' property, i.e., it is the most stable for small shocks but the least stable for large shocks (and vice versa for the ring network).  \cite{feinstein2017sensitivity} studies the sensitivity of the Eisenberg-Noe clearing payments w.r.t.\ the relative liability matrix; that work uses these sensitivities in order to find the best and worst case \emph{directions} for rerouting of the liability network.  \cite{CCY16} utilizes majorization of the financial networks in order to guarantee the relative health of two financial networks (w.r.t.\ the number of defaulting banks). The network compression is also related to the literature on analyzing the consequences of different netting mechanisms in centrally cleared financial markets; see, e.g.,~\cite{duf_zhu_11, duffie2015central, cont2014central, armenti2017central, amini2020systemic, amini2016fully, moallemi14, capponi2015systemic, amini2020clearing}.

We are primarily motivated in this study by two streams of literature: that of~\cite{AOT15} which proposed the robust fragility of the completely connected network and those of~\cite{d2019compressing,veraart2019does} which proposed frameworks for network compression without considering which levels of compression reduce systemic risk.  In this paper we seek to unify these problems into a single optimization framework, which we call \emph{optimal network compression}.  We further merge these problems with the systemic risk measures of~\cite{chen2013axiomatic,kromer2013systemic} so as to determine the compressed network that minimizes systemic risk; as shown in~\cite{veraart2019does}, network compression need not improve systemic risk.

\subsection{Primary Contributions and Organization}\label{sec:intro-contribution}

The primary innovations and results of this paper are in multiple directions.  
First, we introduce a general formulation of the optimal network compression problem for financial systems in the specific context of systemic risk measurement.  This general formulation is presented for different levels of network compression or rerouting allowed from the initial interbank network.  
Furthermore, we prove that this optimal network compression problem is generically NP-hard.  This motivates us to consider a machine learning approach to approximating the optimal financial network. In particular, we use a genetic algorithm and show that, in a simple example of three-bank system, the algorithm performs as well as the optimal network found using an interior point algorithm for the rerouting problem and non-conservative compression.

As a second contribution, we study the (sub)optimality of the maximally compressed network when considered through the lens of systemic risk measure objectives. The maximal compression problem, as studied in~\cite{d2019compressing}, focuses on removing the maximal amount of liabilities in the system subject to certain financial constraints.  Under the typical compression constraints, this problem is a linear programming problem and, thus, can be solved in polynomial time.  
We propose simple metrics to test the (sub)optimality of the maximally compressed network in general stochastic settings. We show that maximal and optimal compression coincide so long as the system is secure enough.  Our numerical case studies illustrates that, under larger stresses, the maximal compression may rapidly become suboptimal.

Third, in order to find tractable analytical forms for the systemic risk measures, as in~\cite{BF18comonotonic}, we consider stress scenarios under systematic shocks. These analytical forms are novel extension of the formulas provided in~\cite{BF18comonotonic} to the systemic risk measures of~\cite{chen2013axiomatic,kromer2013systemic}. Such systematic stress scenarios allow us numerically implement the optimal compression problem tractably. We implement this framework in two case studies.  First, we extend the work of, e.g.,~\cite{AOT15} to consider the robustness of various network topologies for heterogeneous networks under systematic shocks.  In particular, in numerical experiments, we show that under systematic shocks the robust fragility results of~\cite{AOT15} no longer hold generally.  
Second, we consider a financial system calibrated to the 2011 European Banking Authority stress testing data to demonstrate numerical results for a large financial system. %We show that if the systemic risk measures are not explicitly part of the objective function of the optimal compression problem, we may end up with suboptimal outcomes or, even worse, an increase in systemic risk. We show that the optimal network compression can find significant systemic risk improvements over the original network. This framework allows us to numerically generalize~\cite{veraart2019does} to quantify the suboptimality of the full compression as utilized in~\cite{d2019compressing}.
%\end{itemize}

The organization of this paper is as follows.  In Section~\ref{sec:optimization}, we propose the general optimal network compression problem to be considered throughout this work.  In so doing, we formulate four meaningful types of ``compression'' problems motivated by~\cite{d2019compressing,AOT15} and prove that the optimal compression problem is, generically, NP-hard given each of those constraint sets.  We then focus on a meaningful form for the objective function, namely the systemic risk measures of~\cite{chen2013axiomatic,kromer2013systemic}, in Section~\ref{sec:objective}.  Specifically, in Appendix~\ref{sec:systematic}, we present analytical forms for specific examples of these systemic risk measures under systematic shocks to the financial system.  
Within Section~\ref{sec:maximal}, we study the ``maximal'' compression algorithm (i.e., as proposed in~\cite{d2019compressing}) in order to study conditions on the optimality of this compressed network.
This is followed by two case studies in Section~\ref{sec:casestudy}.  First, we consider a simple three bank system that serves the dual purposes of validating our algorithmic approach as well as testing the results of~\cite{AOT15} in a heterogeneous system with systematic shocks.  Second, we study a network calibrated to the 2011 European Banking Authority dataset; with this network we study the improvements in systemic risk via optimal compression and maximal compression %\strike{(i.e., as proposed in}~\rd{\cite{d2019compressing}}\strike{)} 
on a large, realistic financial system.  Section~\ref{sec:conclusion} concludes.

\section{The optimization problem}\label{sec:optimization}

Throughout this work we will consider a system of $n$ banks with obligations $L_{ij} \geq 0$ from bank $i$ to $j$; as is typical, we assume $L_{ii} = 0$ for every bank $i$, i.e., no bank has any obligations to itself.  Additionally, each bank $i$ will be assumed to have liabilities external to the banking network $L_{i0} \geq 0$.  These external obligations are sometimes called societal obligations; we will interchangeably use these terms throughout this work. The set of all such networks is denoted by $\lcal := \{L \in \bbr^{n \times (n+1)}_+ \; | \; L_{ii} = 0 \; \forall i\}$.

In this section we present the primary optimization problem of interest in this work.  To do so we introduce the notion of portfolio compression which we take from~\cite{d2019compressing}.  
In this work we seek to find the optimal network compression problem, i.e., for an initial liability network $\tilde L \in \lcal$, we wish to minimize some objective $f: \lcal \to \bbr$ subject to compression constraints $\ccal: \lcal \to 2^\lcal$ (which are formally presented in Definition~\ref{defn:compression} below).
\begin{equation}
\label{eq:optimization} \min\left\{f(L) \; | \; L \in \ccal(\tilde L)\right\}.
\end{equation}
The objective function $f$ can be directly computed from the network statistics (e.g., network entropy) or the results of systemic risk measures (see, e.g.,~\cite{chen2013axiomatic,kromer2013systemic}).  Note that  in general this function may also depend on the initial liability network $f(L)=f(L;\tilde{L})$, but we drop this from the notation for compactness. 
We will discuss the systemic risk measure based objective functions in the following sections.  Notably, as these objective functions are in general nonconvex, this optimization problem might be hard; in fact, we will show that (generically) this problem is NP-hard given certain objectives and network compression based constraints in Theorem~\ref{thm:NP}.

\begin{remark}\label{rem:max-compress}
Prior works on network compression, e.g.,~\cite{d2019compressing}, focus on removing the maximal amount of liabilities in the system subject to certain financial constraints.  That is, with objective $f(L) := \sum_{i = 1}^n \sum_{j = 0}^n L_{ij}$.  In particular, with the compression constraints highlighted within this work, \eqref{eq:optimization} becomes a linear programming problem and can be solved in polynomial time.  Other, non-optimization based, algorithms for undertaking this compression are presented in~\cite{d2019compressing}.  In contrast, we are motivated, as in~\cite{veraart2019does}, to study partial compression to determine the optimal level of compression; \cite{veraart2019does} focuses on conservative compression (defined below) with a strict definition for optimality related to the set of defaulting banks.
\end{remark}

The constraint set $\ccal(\tilde L)$ denotes the set of all possibly compressed or rerouted networks consistent with $\tilde L$.  Any such meaningful compression problem satisfies two properties: consistent net liabilities and feasibility as a network.  This is encoded in the following definition consistent with prior works on network compression, e.g.,~\cite{d2019compressing,veraart2019does}.
\begin{definition}\label{defn:compression}
Given an initial financial liability network $\tilde L \in \lcal$, $\ccal(\tilde L)$ is a \textbf{set of compressed networks} if $L \in \ccal(\tilde L)$ implies:
\begin{itemize}
\item \textbf{constant net liabilities}: $\sum_{j = 1}^n [L_{ij} - L_{ji}] + L_{i0} = \sum_{j = 1}^n [\tilde L_{ij} - \tilde L_{ji}] + \tilde L_{i0}$ for every bank $i \in \{1, \dots, n\}$ and
\item \textbf{feasibility}: $L \in \lcal$.
\end{itemize}
\end{definition}
Compare the definition of the set of compressed networks to the General Compression Problem defined in~\cite{d2019compressing}.  We wish to note that the set of compressed networks can often be defined as a convex polyhedron; in fact it is explicitly defined this way in~\cite{d2019compressing} and every specific example we consider in this work  satisfy such a structure.

\begin{example}\label{ex:compression}
In this example we consider 3 types of network compression and, fourth, the rerouting problem, which we will consider throughout this work.  These 3 network compression problems are detailed in~\cite{d2019compressing} with conservative compression studied further in~\cite{veraart2019does}.  We describe these compression problems in varying order from most to least restrictive; though each is financially meaningful, other types of compression can be implemented.  To simplify notation, we will define $L_{0i} := 0$ for every bank $i$ for any network $L \in \lcal$.
\begin{enumerate}
\item {\bf Bilateral compression}: Given an initial network $\tilde L \in \lcal$, bilateral compression allows for the reduction of bilateral exposures only.  That is, the net obligations between banks $i$ and $j$ must always be kept consistent with the initial network construction; additionally, all obligations can only be reduced from their initial levels.  As such, we can define bilateral compression $\ccal^B(\tilde L)$ as:
\begin{align*}
\ccal^B(\tilde L) &:= \left\{L \in \lcal \; | \; \forall i,j: \; L_{ij} - L_{ji} = \tilde L_{ij} - \tilde L_{ji}, \; L_{ij} \in [0,\tilde L_{ij}]\right\}.
\end{align*}
Bilateral compression is special insofar as the most the network can be compressed in this way is defined by obligations: $L_{ij} := \max\{0 , \tilde L_{ij} - \tilde L_{ji}\}$. 
An example of bilateral compression is shown in Figure~\ref{fig:ExBilat}.

\begin{figure}

\centering
\begin{tikzpicture}[scale=\textwidth/15cm, >=stealth',shorten >=1pt,node distance=4cm,on grid,initial/.style    ={}]

  \node[state] (x1) at (0,0) {\footnotesize Bank 1};
 \node[state] (x2)  at (5,0) (x2) {\footnotesize Bank 2};
 \node[state] (x3)  at (2.5,4.33) {\footnotesize Bank 3};

\tikzset{every node/.style={fill=white}}

      \tikzset{mystyle/.style={->,relative=true,in=165,out=15,double=blue}}
      
        \path[->] (x1) edge [mystyle]   node {1} (x2);
          \path[->] (x2) edge [mystyle]   node {10} (x1);
          
           \tikzset{mystyle/.style={->,relative=true,in=165,out=15,double=red}}
  
    \path[->] (x2) edge [mystyle]   node {2} (x3);
    
      \path[->] (x3) edge [mystyle]   node {20} (x2);

            \tikzset{mystyle/.style={->,relative=true,in=165,out=15,double=black}}
            
           \path[->] (x3) edge [mystyle]   node {3} (x1);
    
      \path[->] (x1) edge [mystyle]   node {30} (x3);

  \node[state] (x1) at (7,0) {\footnotesize Bank 1};
 \node[state] (x2)  at (12,0) (x2) {\footnotesize Bank 2};
 \node[state] (x3)  at (9.5,4.33) {\footnotesize Bank 3};

\tikzset{every node/.style={fill=white}}

            \tikzset{mystyle/.style={->,relative=true,in=165,out=15,double=blue}}
      
        \path[->] (x1) edge [mystyle]   node {$1-\mu_1$} (x2);
          \path[->] (x2) edge [mystyle]   node {$10-\mu_1$} (x1);
          
         \tikzset{mystyle/.style={->,relative=true,in=165,out=15,double=red}}
  
    \path[->] (x2) edge [mystyle]   node {$2-\mu_2$} (x3);
    
      \path[->] (x3) edge [mystyle]   node {$20-\mu_2$} (x2);

            \tikzset{mystyle/.style={->,relative=true,in=165,out=15,double=black}}
            
           \path[->] (x3) edge [mystyle]   node {$3-\mu_3$} (x1);
    
      \path[->] (x1) edge [mystyle]   node {$30-\mu_3$} (x3);
      
      \end{tikzpicture}

\caption{Example of bilateral compression with $\mu_1 \in [0,1], \mu_2\in[0,2]$ and $\mu_3\in[0,3]$.}
\label{fig:ExBilat}
\end{figure}
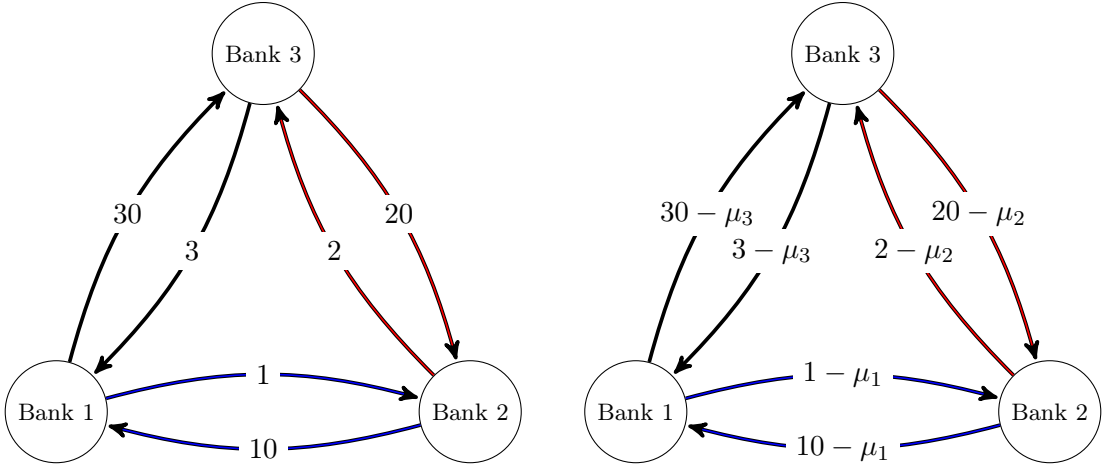

\item {\bf Conservative compression}: Given an initial network $\tilde L \in \lcal$, conservative compression allows for the reduction of cyclical exposures only.  That is, the net obligations owed around a directed cycle $i \to j_1 \to ... \to j_m \to i$ must always be kept consistent with the initial network construction; additionally, all obligations can only be reduced from their initial levels.  As defined in~\cite{d2019compressing}, this cyclical netting rule can be encoded by the fixed net liabilities condition for every bank $i$.  As such, we can define conservative compression $\ccal^C(\tilde L)$ as:
\begin{align*}
\ccal^C(\tilde L) &:= \left\{L \in \lcal \; | \; \forall i,j: \; \sum_{k = 0}^n [L_{ik} - L_{ki}] = \sum_{k = 0}^n [\tilde L_{ik} - \tilde L_{ki}], \; L_{ij} \in [0,\tilde L_{ij}]\right\}.
\end{align*}
Conservative compression is shown in Figure~\ref{fig:ExComp} under $\alpha \in [0,1]$ and $\beta \in [0,10]$.
\item {\bf Rerouting}: Given an initial network $\tilde L \in \lcal$, rerouting allows for the rewiring of the entire network.  That is, all liabilities are redistributed throughout the system in such a way that net and gross liabilities are kept constant. As such, we can define rerouting $\ccal^R(\tilde L)$ as:
\begin{align*}
\ccal^R(\tilde L) &:= \left\{L \in \lcal \; | \; \forall i,j: \; \sum_{k = 0}^n [L_{ik} - L_{ki}] = \sum_{k = 0}^n [\tilde L_{ik} - \tilde L_{ki}], \; \sum_{k = 0}^n L_{ik} = \sum_{k = 0}^n \tilde L_{ik}\right\}.
\end{align*}
Rerouting is shown in Figure~\ref{fig:ExComp} under $\alpha \in (-\infty,1]$, $\beta \in (-\infty,10]$ and $\alpha + \beta = 0$.
\item {\bf Nonconservative compression}: Given an initial network $\tilde L \in \lcal$, nonconservative compression allows for, first, rerouting of the network and, second, the conservative compression of that network.  That is, for every bank $i$, net liabilities are kept constant while gross liabilities are allowed to be reduced from the initial setup.  As such, we can define nonconservative compression $\ccal^N(\tilde L)$ as:
\begin{align*}
\ccal^N(\tilde L) &:= \left\{L \in \lcal \; | \; \forall i,j: \; \sum_{k = 0}^n [L_{ik} - L_{ki}] = \sum_{k = 0}^n [\tilde L_{ik} - \tilde L_{ki}], \; \sum_{k = 0}^n L_{ik} \leq \sum_{k = 0}^n \tilde L_{ik}\right\}.
\end{align*}
\end{enumerate}
Sometimes one may also wish to fix the obligations $L_{i0}$ owed to society; this is accomplished by taking the intersection $\ccal(\tilde L) \cap \ccal^0(\tilde L)$ where
\[\ccal^0(\tilde L) := \{L \in \bbr^{n \times (n+1)}_+ \; | \; \forall i: \; L_{i0} = \tilde L_{i0}\}.\]
We denote the \textbf{nonconservative-0 compression} as $\ccal^N(\tilde L) \cap \ccal^0(\tilde L)$.
Nonconservative compression is shown in Figure~\ref{fig:ExComp} under $\alpha \in (-\infty,1]$, $\beta \in (-\infty,10]$ and $\alpha + \beta \geq 0$.
\end{example}

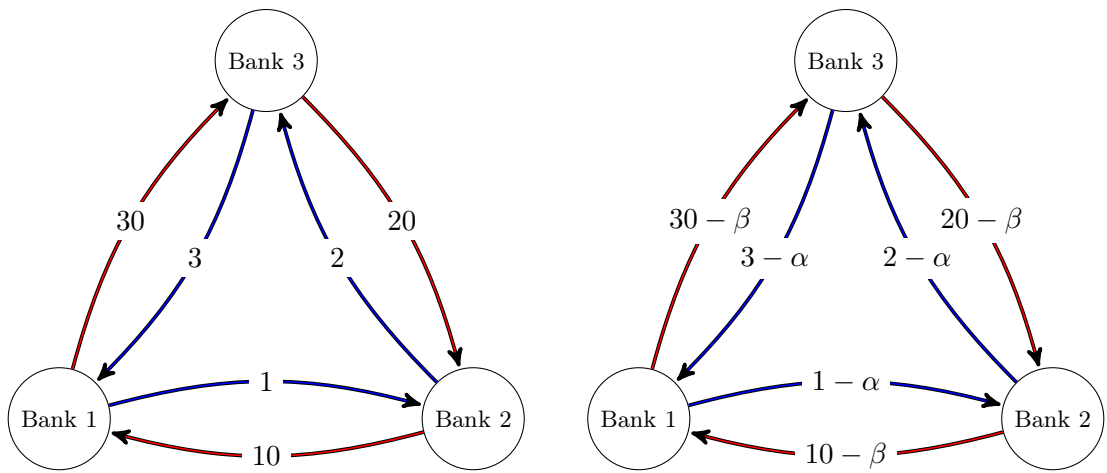
\begin{figure}

\centering
\begin{tikzpicture}[scale=\textwidth/15cm, >=stealth',shorten >=1pt,node distance=4cm,on grid,initial/.style    ={}]

  \node[state] (x1) at (0,0) {\footnotesize Bank 1};
 \node[state] (x2)  at (5,0) (x2) {\footnotesize Bank 2};
 \node[state] (x3)  at (2.5,4.33) {\footnotesize Bank 3};

\tikzset{every node/.style={fill=white}}

      \tikzset{mystyle/.style={->,relative=true,in=165,out=15,double=blue}}
      
        \path[->] (x1) edge [mystyle]   node {1} (x2);
        
         \path[->] (x2) edge [mystyle]   node {2} (x3);
         
            \path[->] (x3) edge [mystyle]   node {3} (x1);

           \tikzset{mystyle/.style={->,relative=true,in=165,out=15,double=red}}
  
            \path[->] (x2) edge [mystyle]   node {10} (x1);
   
      \path[->] (x3) edge [mystyle]   node {20} (x2);

      \path[->] (x1) edge [mystyle]   node {30} (x3);

  \node[state] (x1) at (7,0) {\footnotesize Bank 1};
 \node[state] (x2)  at (12,0) (x2) {\footnotesize Bank 2};
 \node[state] (x3)  at (9.5,4.33) {\footnotesize Bank 3};

\tikzset{every node/.style={fill=white}}

      \tikzset{mystyle/.style={->,relative=true,in=165,out=15,double=blue}}
      
        \path[->] (x1) edge [mystyle]   node {$1-\alpha$} (x2);
        
         \path[->] (x2) edge [mystyle]   node {$2-\alpha$} (x3);
         
            \path[->] (x3) edge [mystyle]   node {$3-\alpha$} (x1);

           \tikzset{mystyle/.style={->,relative=true,in=165,out=15,double=red}}
  
            \path[->] (x2) edge [mystyle]   node {$10-\beta$} (x1);
   
      \path[->] (x3) edge [mystyle]   node {$20-\beta$} (x2);

      \path[->] (x1) edge [mystyle]   node {$30-\beta$} (x3);
            
      \end{tikzpicture}

\caption{Example of cycle compression in different models. For conservative compression: $\alpha\in[0,1]$ and $\beta\in[0,10]$. For rerouting: $\alpha\in(-\infty, 1], \beta \in (-\infty, 10]$ and $\alpha+\beta=0$. For nonconservative compression:  $\alpha\in(-\infty, 1], \beta \in (-\infty, 10]$ and $\alpha+\beta\geq 0$.}
\label{fig:ExComp}
\end{figure}

We conclude this section by showing that the optimal compression problem for the constraint sets of Example~\ref{ex:compression} are NP-hard in general.  
This result motivates us to consider specific settings and algorithms used later in this work.

\begin{theorem}\label{thm:NP}
The optimal network compression problem is NP-hard for the conservative, rerouting and  nonconservative compression models. 
\end{theorem}

The proof of this theorem is provided in Appendix~\ref{sec:proofNP}. By considering the network (relative liabilities) entropy ($f(L) = -\sum_{i = 1}^n \sum_{j = 0}^n \frac{L_{ij}}{\sum_{k = 0}^n L_{ik}}\log\left(\frac{L_{ij}}{\sum_{k = 0}^n L_{ik}}\right)$)  as our objective function we show that the optimal compression problem for each set of constraints is NP-hard.  We prove this by performing reduction from instances of  the NP-complete \textsc{subset sum} problem~\cite{karp1972reducibility}; defined by a set of positive integers $S=\{k_1, k_2, \dots, k_n\}$ and an integer target value $\theta \in \N$, we wish to know whether there exist a subset of these integers that sums up to $\theta$. We show that this can be viewed as a special case for the optimal network compression  for each set of constraints. 

\begin{remark}\label{rem:entropy}
In contrast to using the maximum entropy to find the missing liabilities as in~\cite{M11}, we can consider the minimum entropy as our objective function for compression, see e.g.~\cite{kovavcevic2012hardness, watanabe1981pattern}. Indeed, the network  maximizing the entropy will be close to the complete regular network. On the other hand, the network minimizing the entropy would correspond to a sparse network. So it makes sense to consider it as an objective function for compression; as this is shown in~\cite{watanabe1981pattern}, many of the known algorithms in pattern recognition can be characterized as efforts to minimize the entropy. Consider for example the case of similar firms all having the same total assets and liabilities (and without obligations to society). Then it is easy to check that the network maximizing entropy will correspond to the complete regular network with $\frac{L_{ij}}{\sum_{k = 1}^n L_{ik}}=\frac{1}{n-1}$ which gives $f(L) = n \log(n-1)$. On the other hand, the network minimizing the entropy would correspond to the (regular) ring which corresponds to $f(L)=0$.
\end{remark}

\begin{remark}\label{rem:genetic}
As the optimal compression problem~\eqref{eq:optimization} is generically nonconvex and NP-hard, we cannot rely on a gradient descent method to converge to the global optimum.  As such we believe that machine learning tools and methods would be best for solving such problems in general.  For this paper, as will be utilized and validated in Section~\ref{sec:casestudy}, we will implement a genetic algorithm to solve the optimal compression problem; an overview of the genetic algorithm and its implementation is provided in Appendix~\ref{sec:gp}.\footnote{Briefly, all implementations of the genetic algorithm are completed using the ``ga'' function in the Global Optimization Toolbox of MATLAB.} We direct the reader to Figures~\ref{fig:acemoglu-opt-reroute} and~\ref{fig:acemoglu-opt-compress} for the validation of the genetic algorithm for the optimal compression and rerouting problems in a simple three bank setting.
\end{remark}

\section{Systemic risk objective}\label{sec:objective}
In this section we wish to give a specific structure to the objective function $f: \lcal \to \bbr$ in our optimal compression problem~\eqref{eq:optimization}.  Specifically, we wish to consider the network compression that minimizes a \emph{systemic risk measure}.  These functions are decomposed as $\rho \circ \Lambda$ for a risk measure $\rho$ and an aggregation function $\Lambda$.  Such functions were first introduced in~\cite{chen2013axiomatic,kromer2013systemic} and are detailed below; these mappings also coincide with the ``insensitive systemic risk measures'' of~\cite{feinstein2014measures,AR16}. All results presented within this section are provided solely to provide background information on these systemic risk measures; we believe these mappings provide an important class of objective functions for the optimal compression problem.  These systemic risk measures are utilized within Sections~\ref{sec:maximal} and~\ref{sec:casestudy} below.  In order to present this setting, and for the remainder of this paper, we fix some probability space $(\Omega,\fcal,\bbp)$.  Let $\LL^2 := \LL^2(\Omega,\fcal,\bbp)$ denote those random variables that are square-integrable.

%\subsection{General stress scenario}\label{sec:general}
In order to determine the health of a financial network, we first present a generic aggregation function $\Lambda$ in the following definition.  These aggregation functions are mappings of two arguments: the endowment for the banks and the liability network.  The purpose of such a function is to provide an aggregate statistic of the state of the financial system.
\begin{definition}\label{defn:metric}
The mapping $\Lambda: \bbr^n_+ \times \lcal \to \bbr$ is a \textbf{aggregation function} if it is a nondecreasing mapping in its first argument.
\end{definition}
\begin{example}\label{ex:agg}
Throughout this work we specifically consider four different aggregation functions that are all fundamentally associated with the clearing mechanisms of~\cite{EN01,RV13} to incorporate initial margins and collateral similar to that presented in, e.g.,~\cite{ghamami2021collateralized,veraart2019does}.   
That is, for endowments $x \in \bbr^n_+$, margin $\mu \in [0,1]$, and recovery rates $\alpha_x,\alpha_L \in [0,1]$, the clearing payments are the \emph{maximal} fixed point $p(x,L) = \FIX_{p \in [0,L\vec{1}]} \Psi(p;x,L)$ for 
\begin{align*}
%\Psi_i(p;x,L) &= \mu\sum\limits_{j = 0}^n L_{ij} + \begin{cases} (1-\mu)\sum\limits_{j = 0}^n L_{ij} &\text{if } x_i + \sum\limits_{j = 1}^n \frac{L_{ji}}{\sum_{k = 0}^n L_{jk}} p_j \geq \strike{(1-\mu)}\sum\limits_{j = 0}^n L_{ij} \\ \left[(1-\mu)\sum\limits_{j = 0}^n L_{ij}\right] \wedge \left[\alpha_x x_i + \alpha_L \sum\limits_{j = 1}^n \frac{L_{ji}}{\sum_{k = 0}^n L_{jk}} p_j\right] &\text{if } x_i + \sum\limits_{j = 1}^n \frac{L_{ji}}{\sum_{k = 0}^n L_{jk}} p_j < \strike{(1-\mu)}\sum\limits_{j = 0}^n L_{ij}. \end{cases}
\Psi_i(p;x,L) &= \begin{cases} \sum\limits_{j = 0}^n L_{ij} &\text{if } x_i + \sum\limits_{j = 1}^n \frac{L_{ji}}{\sum_{k = 0}^n L_{jk}} p_j \geq \sum\limits_{j = 0}^n L_{ij} \\ \left(\sum\limits_{j = 0}^n L_{ij}\right) \wedge \left(\mu\sum\limits_{j = 0}^n L_{ij} + \alpha_x x_i + \alpha_L \sum\limits_{j = 1}^n \frac{L_{ji}}{\sum_{k = 0}^n L_{jk}} p_j\right) &\text{if } x_i + \sum\limits_{j = 1}^n \frac{L_{ji}}{\sum_{k = 0}^n L_{jk}} p_j < \sum\limits_{j = 0}^n L_{ij}. \end{cases}
\end{align*}
As such, the clearing procedure $\Psi$ implies: if bank $i$ has nonnegative wealth $x_i + \sum_{j = 1}^n \frac{L_{ji}}{\sum_{k = 0}^n L_{jk}} p_j - \sum_{j = 0}^n L_{ij} \geq 0$,  then it is solvent and its wealth is equal to its total assets minus its total liabilities; if bank $i$ has negative wealth $x_i + \sum_{j = 1}^n \frac{L_{ji}}{\sum_{k = 0}^n L_{jk}} p_j - \sum_{j = 0}^n L_{ij} < 0$ then it is defaulting and its assets are reduced by the recovery rates $\alpha_x,\alpha_L$  while the collateral is used to cover $100\times\mu\%$ of the liabilities. Due to the use of the collateral, a bank may be able to pay its obligations in full even while defaulting.  We wish to note that we recover the clearing mechanisms of~\cite{EN01,RV13} with this structure when the margin is removed, i.e., $\mu = 0$.  %From~\cite{RV13}, we immediately recover a greatest and least clearing solution to $p = \Psi(p;x,L)$ within the lattice $[0 , L\vec{1}]$.
\begin{proposition}\label{prop:payments}
Fix the margin $\mu \in [0,1]$ and recovery rates $\alpha_x,\alpha_L \in [0,1]$.  There exists a greatest and least clearing solution to $p = \Psi(p;x,L)$ within the lattice $[0 , L\vec{1}]$ for any endowments $x \in \bbr^n_+$ and liability network $L \in \lcal$.
\end{proposition}
\begin{proof}
As $p \in [0,L\vec{1}] \mapsto \Psi(p;x,L)$ is nondecreasing, the result follows trivially from Tarski's fixed point theorem.
\end{proof}

%\strike{Let $V(x,L) \in \bbr^n$ denote the clearing wealths from an Eisenberg-Noe style clearing procedure with endowments $x \in \bbr^n_+$ and $L \in \lcal$ denote the network of obligations, i.e.,}
%\begin{equation}
%\label{eq:wealths} \strike{V_i(x,L) = x_i + \sum_{j = 1}^n \frac{L_{ji}}{\sum_{k = 0}^n L_{jk}} p_j(x,L) - (1-\mu)\sum_{j = 0}^n L_{ij}.}
%\label{eq:wealths} V_i(x,L) = \begin{cases} x_i + \sum_{j = 1}^n \frac{L_{ji}}{\sum_{k = 0}^n L_{jk}} p_j(x,L) - \sum_{j = 0}^n L_{ij} &\text{if } p_i(x,L) = \sum_{j = 0}^n L_{ij} \\ \alpha_x x_i + \alpha_L \sum_{j = 1}^n \frac{L_{ji}}{\sum_{k = 0}^n L_{jk}} p_j(x,L) - \sum_{j = 0}^n L_{ij} &\text{if } p_i(x,L) < \sum_{j = 0}^n L_{ij}. \end{cases}
%\end{equation}

With this clearing procedure, we consider the following four aggregation functions:
\begin{itemize}
\item \textbf{Number of banks paying in full}: $\Lambda^\#(x,L) := \sum_{i = 1}^n \ind{p_i(x,L) \geq \sum_{j = 0}^n L_{ij}}$.
\item \textbf{Number of solvent banks}: $\Lambda^{\scal}(x,L) := \sum_{i = 1}^n \ind{x_i + \sum_{j = 1}^n\frac{L_{ji}}{\sum_{k = 0}^n L_{jk}}p_j(x,L) \geq \sum_{j = 0}^n L_{ij}}$.
\item \textbf{System-wide wealth} : $\Lambda^\ncal(x,L) := \sum_{i = 1}^n \left[x_i + \sum_j \frac{L_{ji}}{\sum_{k = 0}^n L_{jk}} p_j(x,L) - \sum_{j = 0}^n L_{ij}\right]$. 
%\rd{Technically code is with $-(1-\mu)\sum_j L_{ij}$ but only used in 3 bank example with $\mu = 0$.}
\item \textbf{External wealth}: $\Lambda^{0}(x,L) := \sum_{i = 1}^n \frac{L_{i0}}{\sum_{j = 0}^n L_{ij}} p_i(x,L)$.
\end{itemize}
\end{example}

Additionally, we need to consider a risk measure $\rho$ in order to determine the risk that the system is incurring.  Such functions map random variables into capital requirements.
\begin{definition}\label{defn:riskmsr}
The mapping $\rho: \LL^2 \to \bbr$ is a \textbf{risk measure} if it satisfies the following properties:
\begin{itemize}
\item \textbf{normalization}: $\rho(0) = 0$;
\item \textbf{monotonicity}: $\rho(X) \leq \rho(Y)$ if $X \geq Y$ a.s.; and
\item \textbf{translative}: $\rho(X+m) = \rho(X) - m$ for $m \in \bbr$.
\end{itemize}
\end{definition}
These risk measures may satisfy additional conditions, e.g., convexity or positive homogeneity.  
\begin{example}\label{ex:riskmsr}
For the purposes of this work, we will focus on two standard risk measures parameterized by $\gamma \in [0,1]$:
\begin{itemize}
\item \textbf{Value-at-Risk}: $\rho^{\rm VaR}_\gamma(Z) = -\inf\{z \in \bbr \; | \; \bbp(Z \leq z) > 1-\gamma\} =: -Z_{1-\gamma}$.  If $\gamma = 1$ then we recover the so-called \emph{worst-case risk measure}: $\rho^{\rm VaR}_1(Z) = -\essinf Z =: \rho^{\rm WC}(Z)$.
\item \textbf{Expected shortfall}: $\rho^{\rm ES}_\gamma(Z) = -\E[Z | Z \leq Z_{1-\gamma}]$.  If $\gamma = 0$ then we recover the so-called \emph{expectation risk measure}: $\rho^{\rm ES}_0(Z) = -\E[Z] =: \rho^\E(Z)$.
\end{itemize}
\end{example}

To measure the health of the financial system, we consider the systemic risk measures $\rho \circ \Lambda$ for risk measure $\rho$ and aggregation function $\Lambda$ as presented in, e.g.,~\cite{chen2013axiomatic}.
\begin{definition}\label{defn:systemic}
The mapping $R: (\LL^2)^n \times \lcal \to \bbr$ is a \textbf{systemic risk measure} if it can be decomposed into an aggregation function $\Lambda$ and a risk measure $\rho$ so that $R(X,L) := \rho(\Lambda(X,L))$ for every $X \in (\LL^2)^n$ and $L \in \lcal$.
\end{definition}
Within the optimal compression problem~\eqref{eq:optimization}, we specifically are interested in 
\[f(L) := \rho(\Lambda(\xcal_L,L))\]
for some (random) endowments $\xcal: \lcal \to (\LL^2)^n$.  That is, given a stress scenario $\xcal: \lcal \to (\LL^2)^n$, we seek to find the optimal financial network (subject to compression constraints) such that the systemic risk is minimized. In Appendix~\ref{sec:systematic}, a semi-analytic form for these systemic risk measures along the lines as~\cite{BF18comonotonic} is provided under systematic stress scenarios.

\begin{remark}\label{rem:acemoglu}
We wish to highlight two special cases which relate to the notions proposed in~\cite{AOT15} in the uncollateralized $\mu = 0$ setting.  
(In this uncollateralized setting, the number of solvent firms coincides exactly with the number of banks paying in full, i.e., $\Lambda^\# = \Lambda^\scal$.)
Fix $X \in (\LL^2)^n$ such that $X_i \geq 0$ a.s.\ for every bank $i$. In the context of the above general setting, we define the stress scenarios by the random field $\xcal_L \equiv X$ for any $L \in \lcal$.
\begin{enumerate}
\item Consider the worst-case risk measure $\rho^{\rm WC}(Z) = -\essinf Z$.  In the notation from \cite{AOT15}, $L$ is more \textbf{resilient} than $\hat L$ if and only if $$\rho^{\rm WC}(\Lambda^\#(X,L)) \leq \rho^{\rm WC}(\Lambda^\#(X,\hat L)).$$
\item Consider the expectation risk measure $\rho^\E(Z) = -\E[Z]$.  In the notation from \cite{AOT15}, $L$ is more \textbf{stable} than $\hat L$ if and only if $$\rho^\E(\Lambda^\#(X,L)) \leq \rho^\E(\Lambda^\#(X,\hat L)).$$
\end{enumerate}
\cite{AOT15} presents these notions for symmetric systems of banks with i.i.d.\ Bernoulli shocks $X_i$.  Notably, under such conditions, the number of solvent banks provides the full information on the health of the system; that is, any aggregation function that depends on $(x,L)$ only through the clearing payments $p(x,L)$, say $\Lambda$, provides the same ordering of liability networks $L,\hat L \in \lcal$ as $\Lambda^\#$: 
\[\rho(\Lambda(X,L)) \leq \rho(\Lambda(X,\hat L)) \; \Leftrightarrow \; \rho(\Lambda^\#(X,L)) \leq \rho(\Lambda^\#(X,L)).\]  
We will revisit these problems, and compare our formulation with that with~\cite{AOT15} further, in Section~\ref{sec:acemoglu} below.
\end{remark}

\section{Optimality of maximal compression}\label{sec:maximal}
Recall the maximal compression problem as defined in~\cite{d2019compressing} and detailed in Remark~\ref{rem:max-compress}, i.e., the optimal network compression with $f(L) := \sum_{i = 1}^n \sum_{j = 0}^n L_{ij}$. Herein we will focus entirely on compression without consideration for the rerouting problem. Within this section, our goal is to study the (sub)optimality of the maximally compressed network when studied through the lens of systemic risk measure objectives (i.e., $f(L) := R(\xcal_L,L) := \rho(\Lambda(\xcal_L,L))$). Throughout this section, we will denote the maximally compressed network by $L^{\max} := \argmin\{\sum_{i = 1}^n \sum_{j = 0}^n L_{ij} \; | \; L \in \ccal(\tilde L)\}$ for generic convex compression constraints $\ccal(\tilde L)$.  As stated in Remark~\ref{rem:max-compress}, $L^{\max}$ can be computed as the solution to a linear program under any of our proposed compression problems provided within Example~\ref{ex:compression}.
%\rd{Can we come up with sufficient conditions based on $\rho$ and $\Lambda$ so that maximal compression is (i) optimal [unlikely! because even locally optimal does not guarantee global] or (ii) suboptimal [Is the directional derivative of $\rho\circ\Lambda$ from $L^{\max}$ in direction of original network $L$ negative?]}
\begin{remark}\label{rem:init-max}
Consider the systemic risk objective $f(L) := R(\xcal_L,L)$ presented above.  If the only concern is the relative risk of the original network to the maximally compressed network, this can be compared directly by studying $R(\xcal_{\tilde L},{\tilde L})$ and $R(\xcal_{L^{\max}},L^{\max})$.  Specifically, if $R(\xcal_{L^{\max}},L^{\max}) > R(\xcal_{\tilde L},{\tilde L})$ then maximal compression is worse than no compression and, as a direct consequence, must also be suboptimal for the optimal network compression problem~\eqref{eq:optimization}.
\end{remark}
\begin{proposition}\label{prop:derivative}
Consider a financial network $\tilde L \in \lcal$ with maximally compressed version $L^{\max} \in \lcal$ under compression constraints $\ccal(\tilde L)$.
Additionally, consider the optimal network compression problem over systemic risk measure $R: (\LL^2)^n \times \lcal \to \bbr$ subject to the stress scenario $\xcal: \lcal \to (\LL^2)^n$.  
Assume $L \in \lcal \mapsto R(\xcal_L,L)$ is differentiable.
Denote the directional derivative of $R(\xcal_{\cdot},\cdot)$ at $L^{\max}$ in the direction of the initial liabilities $\tilde L$ by
\[\dcal := %D_{\tilde L - L^{\max}} \rho \circ \Lambda(\xcal_{L^{\max}},L^{\max}) = 
    \lim_{t \searrow 0} \frac{1}{t}\left[R(\xcal_{t\tilde L + (1-t)L^{\max}},t\tilde L + (1-t)L^{\max}) - R(\xcal_{L^{\max}},L^{\max})\right].\]  
If $\dcal < 0$ then $L^{\max}$ is suboptimal, i.e., there exists some $L^* \in \lcal$ such that $R(\xcal_{L^*},L^*) < R(\xcal_{L^{\max}},L^{\max})$.  
If $L^{\max}$ is the optimally compressed network then $\dcal \geq 0$.
\end{proposition}
\begin{proof}
If $\dcal < 0$ then there exists some $t > 0$ such that $R(\xcal_{t\tilde L + (1-t)L^{\max}},t\tilde L + (1-t)L^{\max}) < R(\xcal_{L^{\max}},L^{\max})$, i.e., $L^{\max}$ is suboptimal by construction.
In contrast, if $L^{\max}$ is the optimally compressed network then, by definition of optimality, the directional derivative of $R(\xcal_{\cdot},\cdot)$ at $L^{\max}$ is nonnegative in every direction; in particular, $\dcal \geq 0$.
\end{proof}
In Remark~\ref{rem:init-max} and Proposition~\ref{prop:derivative}, we propose simple metrics to test the (sub)optimality of the maximally compressed network $L^{\max}$ in general stochastic settings.  We wish to conclude this discussion of the optimality of the maximally compressed network under a specific \emph{systematic} structure for the stress scenario.  (As noted above, a semi-analytic form for the Value-at-Risk and expected shortfall under the aggregation functions of Example~\ref{ex:agg} are provided within Appendix~\ref{sec:systematic}.) In contrast to, e.g.,~\cite{BF18comonotonic}, this result permits a comparison of network structures in the vein of~\cite{AOT15}. As far as the authors are aware, such a result is completely novel in the literature.
\begin{proposition}\label{prop:max}
Consider a financial network $\tilde L \in \lcal$ with maximally compressed version $L^{\max} \in \lcal$ under compression constraints $\ccal(\tilde L)$.
Additionally, consider the optimal network compression problem over systemic risk measure $R: (\LL^2)^n \times \lcal \to \bbr$ with decomposition $R := \rho \circ \Lambda$ such that $\rho$ is a law-invariant risk measure and $\Lambda$ is an aggregation function.  Assume $\Lambda$ is based on the margined clearing mechanism of~\cite{EN01,RV13} as provided within Example~\ref{ex:agg} with $\mu,\alpha_x,\alpha_L \in [0,1]$ such that $\alpha_L \leq 1-\mu$ and $\Lambda(x,L) \geq \Lambda(\bar x,\bar L)$ if a larger fraction of obligations are paid under $(x,L)$ than $(\bar x,\bar L)$, i.e., if $p_i(x,L)/\sum_{j = 0}^n L_{ij} \geq p_i(\bar x,\bar L)/\sum_{j = 0}^n \bar L_{ij}$ for every bank $i$.
Finally, assume a parameterized stress scenario $\xcal_L^q$ for $q \in \LL^2$ such that $\xcal_L^q \geq \xcal_{\bar L}^{\bar q}$ (w.r.t.\ first order stochastic dominance) if $\sum_{j = 0}^n L_{ij} \leq \sum_{j = 0}^n \bar L_{ij}$ and $q \geq \bar q$ (w.r.t.\ first order stochastic dominance).
%\rd{Finally, assume the stress scenarios follow $\xcal^{\mu,q}_L := (b + \phi\mu \times [\sum_{j = 0}^n (\tilde L_{ij} - L_{ij})]) + (s + (1-\phi)\mu \times [\sum_{j = 0}^n (\tilde L_{ij} - L_{ij})])q$ for nonnegative random shock $q \in \LL^2$.}
%\rd{$q \mapsto R(\xcal_L^{\mu,q},L) - R(\xcal_{L^{\max}}^{\mu,q},L^{\max})$ is nondecreasing (w.r.t.\ first-order stochastic dominance) for any feasible network $L \in \ccal(\tilde L)$.}
Then there exists some $q^* \in \LL^2$ such that $L^{\max} \in \argmin\{R(\xcal_L^q,L) \; | \; L \in \ccal(\tilde L)\}$ for every $q \geq q^*$ (w.r.t.\ first order stochastic dominance). %, i.e., maximal and optimal compression coincide so long as system is secure enough. 
\end{proposition}
\begin{proof}
First, we wish to note that if $p(\xcal_L^q,L) = L\vec{1}$ a.s.\ then $p(\xcal_{L^{\max}}^q,L^{\max}) = L^{\max}\vec{1}$ as well; this holds because either $\xcal_{L^{\max}}^q \geq \xcal_L^q \geq \sum_{j = 0}^n [L_{ij} - L_{ji}] = \sum_{j = 0}^n [L_{ij}^{\max} - L_{ji}^{\max}]$ (i.e., no bank is defaulting under the maximally compressed network $L^{\max}$) or $\alpha_x \xcal_{L^{\max}}^q - \alpha_L \sum_{j = 0}^n [L_{ij}^{\max} - L_{ji}^{\max}] - (1-\mu-\alpha_L)\sum_{j = 0}^n L_{ij}^{\max} \geq \alpha_x \xcal_L^q - \alpha_L \sum_{j = 0}^n [L_{ij} - L_{ji}] - (1-\mu-\alpha_L)\sum_{j = 0}^n L_{ij} \geq 0$ (i.e., if a bank is defaulting then it must still pay all obligations in full due to the margin requirements).   
Take $$q^* := \inf\{q \in \bbr \; | \; p_i(\xcal_{L^{\max}}^q,L^{\max}) = \sum_{j = 0}^n L_{ij}^{\max} \; a.s.\ \forall i\},$$ then for any $q \geq q^*$ it must hold that $R(\xcal_{L^{\max}}^q,L^{\max}) = \min\{R(\xcal_L^q,L) \; | \; L \in \ccal(\tilde L)\}$ by monotonicity of the aggregation function $\Lambda$.
\end{proof}
Proposition~\ref{prop:max} means that maximal and optimal compression coincide so long as the system is secure enough.  Heuristically, under larger stresses, the maximal compression may rapidly become suboptimal as will be observed in the numerical case studies in Section~\ref{sec:casestudy} below. 
\begin{remark}\label{rem:max-conditions}
Though Proposition~\ref{prop:max} introduces a number of conditions, these are naturally satisfied in many cases.  
Notably, every aggregation function introduced within Example~\ref{ex:agg} satisfies the monotonicity with respect to the fraction of obligations being repayed.
Further, the parametric form for the stress scenarios $\xcal_L^q$ with $q \in \LL^2$ is taken to impose a systematic shock $q$ on the assets of the banks. This special structure is utilized throughout the case studies constructed within Section~\ref{sec:casestudy} below. Assuming $\mu > 0$, the monotonicity of the stress scenarios in the liabilities matrices is because, in the case studies, the margin is returned as liabilities are returned (and the shocks are limited to nonnegative random variables).
\end{remark}

\section{Case studies}\label{sec:casestudy}
In this section we will consider two case studies to demonstrate the results of this work.  As mentioned in Remark~\ref{rem:genetic}, we implement a genetic algorithm to optimize~\eqref{eq:optimization}; an overview of the genetic algorithm is provided in Appendix~\ref{sec:gp}.

First, we will present a small, 3 bank, financial system with heterogeneous (comonotonic) endowments.  This system allows us to easily present analytical results to validate the genetic algorithm we utilize to consider optimal compression and rerouting.  Additionally, this small system allows us to investigate the robust fragility results of \cite{AOT15} in a different setting to determine if those results hold for more general settings than presented in that work. In particular, we show that under systematic shocks, even for a simple heterogeneous three bank system, the robust fragility results of~\cite{AOT15} no longer hold generally.

Second, we calibrate a financial system to the 2011 European Banking Authority stress testing data.  With that system we compare the original network, the fully compressed networks, and optimally compressed networks.  
%\strike{As found in \cite{veraart2019does}, we find that network compression can increase systemic risk, but optimal network compression can find significant improvements over the original network.} 
This framework allows us to numerically generalize~\cite{veraart2019does} to quantify the suboptimality of the full compression as utilized in~\cite{d2019compressing} on a large financial network calibrated to data in a specific example. Namely, we show that if the systemic risk measures are not explicitly part of the objective function of the optimal compression problem, we may end up with suboptimal outcomes or, even worse, increased systemic risk. We also find that the optimal network compression can find significant systemic risk improvements over the original network.

Broadly, these numerical examples demonstrate: (i) the optimal compression problem is computationally feasible; (ii) robust fragility with i.i.d.\ Bernoulli shocks as reported in~\cite{AOT15} does not hold for heterogeneous systems; and (iii) maximal compression can be suboptimal in large ``realistic'' financial networks.  As such, these case studies should be viewed through that lens, i.e.\ as numerical counterexamples to the widely reported robust fragility (Section~\ref{sec:acemoglu}) and the optimality of maximal compression in large, complex financial systems (Section~\ref{sec:EBA}).
With these goals in mind, we want to emphasize that all networks considered herein are fully ``nettable'' in the sense that maximal conservative compression can net out the entire interbank network.  Such a simple network structure was taken for demonstration purposes of the three noted goals only.  Further studies with different network calibrations could be of interest for future works.

\subsection{Three bank system}\label{sec:acemoglu}
In~\cite{AOT15} comparison of completely connected to ring structured networks was undertaken in an uncollateralized setting with fixed (random) stress scenarios.  In that work, these networks were ``symmetric'' in that all banks were identical in assets and liabilities, but with i.i.d.\ Bernoulli shocks.  That work focuses on stability and resilience, i.e., w.r.t.\ $\rho^{\rm WC} \circ \Lambda^\#$ and $\rho^\E \circ \Lambda^\#$, respectively, as provided in Remark~\ref{rem:acemoglu}. Notably,~\cite{AOT15} determines, under i.i.d.\ Bernoulli shocks, that the dense network is more stable if shocks are small, but the sparse network is more stable if the shocks are large.  We will explore this question further with a consideration of a 3 bank system under \emph{systematic} shocks that allows for: a completely connected network and two ring networks. The general network structure is depicted in Figure~\ref{fig:acemoglu-network}.  For the sake of notational simplicity, we will consistently refer to bank $i \in \{1,2,3\}$ as the bank with endowment equal to $x_i \times q$.  Without loss of generality we assume $x_1 \leq x_2 \leq x_3$.  We additionally assume that either $x_1 \neq x_2$, $x_1 \neq x_3$, or $x_2 \neq x_3$; if $x_1 = x_2 = x_3$ then, due to the comonotonic endowments inherent for a systematic shock, all rerouted networks have identical systemic risk. We wish to note that both the net and gross liabilities of each bank are the same in all network setups.  For direct comparison with~\cite{AOT15}, within this section we will focus entirely on the uncollateralized setting (i.e., $\mu = 0$) with $\xcal_L \equiv x \times q$ for any network $L \in \lcal$. Notably, as mentioned in Remark~\ref{rem:acemoglu}, in this uncollateralized setting the number of solvent firms coincides exactly with the number of banks paying in full, i.e., $\Lambda^\# = \Lambda^\scal$; because of this, we will not consider $\Lambda^\scal$ explicitly herein.

\begin{figure}
\centering
\begin{tikzpicture}[scale=\textwidth/10cm, >=stealth',shorten >=1pt,node distance=4cm,on grid,initial/.style    ={}]

  \node[state] (x1) at (0,0) {$x_1q$};
 \node[state] (x2)  at (5,0) (x2) {$x_2q$};
 \node[state] (x3)  at (2.5,4.33) {$x_3q$};
 \node[state] (x0)  at (2.5,1.44) {$x_0$};

       \tikzset{mystyle/.style={->,double=black}} 
\tikzset{every node/.style={fill=white}} 

  \path[->] (x1) edge [mystyle]   node {$y$} (x0);
  
    \path[->] (x2) edge [mystyle]   node {$y$} (x0);
    
      \path[->] (x3) edge [mystyle]   node {$y$} (x0);

      \tikzset{mystyle/.style={->,relative=true,in=165,out=15,double=blue}}
      
        \path[->] (x1) edge [mystyle]   node {$L_{12}$} (x2);
  
    \path[->] (x2) edge [mystyle]   node {$L_{23}$} (x3);
    
      \path[->] (x3) edge [mystyle]   node {$L_{31}$} (x1);

            \tikzset{mystyle/.style={->,relative=true,in=165,out=15,double=red}}
        \path[->] (x2) edge [mystyle]   node {$L_{21}$} (x1);
        
           \path[->] (x1) edge [mystyle]   node {$L_{13}$} (x3);
              \path[->] (x3) edge [mystyle]   node {$L_{32}$} (x2);
      
%\draw[every loop, auto=right, line width=0.5mm, >=latex]
%%(x1) edge[bend right=10] node[below] {$\lambda$} (x2)
%(x1) edge[bend right=10] node[below] {$L_{12}$} (x2)
%%(x2) edge[bend right=10] node[above] {$1-\lambda$} (x1)
%%(x2) edge[bend right=10] node[above] {$\xi$} (x1)
%(x2) edge[bend right=10] node[above] {$L_{21}$} (x1)
%%(x1) edge[bend right=10] node[right] {$1-\lambda$} (x3)
%%(x1) edge[bend right=10] node[right] {$\xi$} (x3)
%(x1) edge[bend right=10] node[right] {$L_{13}$} (x3)
%%(x3) edge[bend right=10] node[left] {$\lambda$} (x1)
%(x3) edge[bend right=10] node[left] {$L_{31}$} (x1)
%%(x2) edge[bend right=10] node[right] {$\lambda$} (x3)
%(x2) edge[bend right=10] node[right] {$L_{23}$} (x3)
%%(x3) edge[bend right=10] node[left] {$1-\lambda$} (x2)
%%(x3) edge[bend right=10] node[left] {$\xi$} (x2)
%(x3) edge[bend right=10] node[left] {$L_{32}$} (x2)
%(x1) edge node[above] {$y$} (x0)
%(x2) edge node[above] {$y$} (x0)
%(x3) edge node[left] {$y$} (x0);
\end{tikzpicture}
\caption{Section~\ref{sec:acemoglu}: The generic network structure under consideration.}
\label{fig:acemoglu-network}
\end{figure}
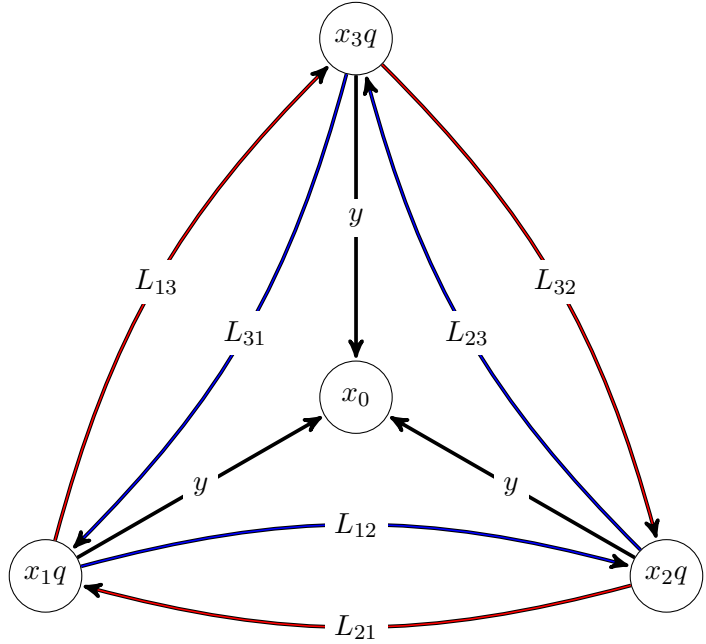

To approach the problem of studying optimal compression and rerouting for our 3 bank system, we first want to study the minimal solvency prices $q^*$ for our three banks.  In order to study these problems we impose the following netting conditions on the obligations:
\[L_{12}+L_{13} = L_{21}+L_{31}, \quad L_{21}+L_{23}=L_{12}+L_{32}, \quad L_{31}+L_{32}=L_{13}+L_{23}.\]  
Given the network provided in Figure~\ref{fig:acemoglu-network} with the aforementioned netting conditions, these values can be computed explicitly as:
\begin{align*}
q_1^* &= \frac{y}{x_1}\\
q_2^* &= \begin{cases} \min\{q_1^* \; , \; \frac{(L_{12}+y)(L_{12}+L_{13}+y) - \alpha_L L_{12}(L_{12}+L_{13})}{\alpha_x L_{12} x_1 + (L_{12}+L_{13}+y)x_2}\} &\text{if } \circled{1}\\
    \min\{q_3^* \; , \; \frac{\bar p_2(\bar p_1 \bar p_3-\alpha_\LL^2 L_{13}L_{31}) - \alpha_L(L_{12}(\alpha_L L_{23}L_{31}+\bar p_3 L_{21}) + L_{32}(\alpha_L L_{21}L_{13}+\bar p_1 L_{23}))}{\alpha_x(\bar p_3 L_{12}+\alpha_L L_{13}L_{32})x_1 + (\bar p_1 \bar p_3 - \alpha_\LL^2 L_{13}L_{31})x_2 + \alpha_x(\bar p_1 L_{32}+\alpha_L L_{31}L_{12})x_3} \} &\text{if } \circled{\cancel{1}} \end{cases}\\
q_3^* &= \begin{cases} \min\{q_2^* \; , \; \frac{\bar p_3(\bar p_1 \bar p_2 - \alpha_\LL^2 L_{12}L_{21}) - \alpha_L(L_{13}(\alpha_L L_{32}L_{21}+\bar p_2 L_{31}) + L_{23}(\alpha_L L_{31}L_{12}+\bar p_1 L_{32}))}{\alpha_x(\bar p_2 L_{13}+\alpha_L L_{12}L_{23})x_1 + \alpha_x(\bar p_1 L_{23}+\alpha_L L_{21}L_{13})x_2 + (\bar p_1 \bar p_2-\alpha_\LL^2 L_{12}L_{21})x_3} \} &\text{if }\circled{1} \\
    \min\{q_1^* \; , \; \frac{(L_{13}+y)(L_{12}+L_{13}+y) - \alpha_L L_{13}(L_{12}+L_{13})}{\alpha_x L_{13} x_1 + (L_{12}+L_{13}+y)x_3}\} &\text{if } \circled{\cancel{1}} \end{cases}
\end{align*}
where $\bar p_1 = L_{12}+L_{13}+y$, $\bar p_2 = L_{21}+L_{23}+y$, $\bar p_3 = L_{31}+L_{32}+y$ and 
\begin{align*}
 \circled{1}:& ([1-\alpha_L][L_{12} + L_{13}] + y)[L_{12} x_3 - L_{13} x_2] + y(L_{12} + L_{13} + y)[x_3 - x_2] + \alpha_x y (L_{13} - L_{12})x_1 \geq 0,\\
\circled{\cancel{1}}:&  ([1-\alpha_L][L_{12} + L_{13}] + y)[L_{12} x_3 - L_{13} x_2] + y(L_{12} + L_{13} + y)[x_3 - x_2] + \alpha_x y (L_{13} - L_{12})x_1 <  0.
\end{align*}
These values $q^*$ allow us to explicitly compute the statistics on the network as discussed in Appendix~\ref{sec:systematic}.

To make these defaulting price levels more explicit, we wish to consider 4 simplified networks with $x_i = i$ and $\bar p_1 = \bar p_2 = \bar p_3 \leq 1$ for all banks.  These networks and resultant $q^*$ are:
\begin{enumerate}
\item \textbf{Completely connected [CC]:} Let $L_{ij} = \frac{1}{2}$ for all $i \neq j \in \{1,2,3\}$ with $y > 0$.
    First we note that, by construction, it must be that $q_1^* \geq q_2^* \geq q_3^*$ in this setup for any choice of bankruptcy costs and obligations to society.  For notation to allow for easier comparisons later on, we will denote these thresholds as $q_1^{CC},q_2^{CC},q_3^{CC}$.
    \begin{align*}
    q_1^{CC} &= y, \\
    q_2^{CC} &= \min\{q_1^{CC} \; , \; \frac{2y^2 + 3y + (1-\alpha_L)}{4y + 4 + \alpha_x}\}, \\
    q_3^{CC} &= \min\{q_2^{CC} \; , \; \frac{2(y^2 + (2-\frac{\alpha_L}{2})y + (1-\alpha_L))}{3(2y + 2 + \alpha_x - \alpha_L}\}. 
    \end{align*}
%    Note that $q_1^{CC} > q_2^{CC}$ if $y > \frac{1}{4}\left(-(1+\alpha_x) + \sqrt{(1+\alpha_x)^2 + 8(1-\alpha_L)}\right)$.

\item \textbf{Ring 123 [123]:} Let $L_{12} = L_{23} = L_{31} = 1$ and $L_{13} = L_{32} = L_{21} = 0$ with $y > 0$.
    First we note that, by construction, it must be that $q_1^* \geq q_2^* \geq q_3^*$ in this setup for any choice of bankruptcy costs and obligations to society.  For notation to allow for easier comparisons later on, we will denote these thresholds as $q_1^{123},q_2^{123},q_3^{123}$.
    \begin{align*}
    q_1^{123} &= y, \\
    q_2^{123} &= \min\{q_1^{123} \; , \; \frac{y^2 + 2y + (1-\alpha_L)}{2y + 2 + \alpha_x}\}, \\
    q_3^{123} &= \min\{q_2^{123} \; , \; \frac{y^3 + 3y^2 + 3y + (1-\alpha_\LL^2)}{3y^2 + (6 + 2\alpha_x)y + (3 + \alpha_x(2+\alpha_L))}\}.
    \end{align*}
%    Note that $q_1^{123} > q_2^{123}$ if $y > \frac{1}{2}\left(-\alpha_x + \sqrt{\alpha_x^2 + 4(1-\alpha_L)}\right)$.

\item \textbf{Ring 132 [132]:} Let $L_{12} = L_{23} = L_{31} = 0$ and $L_{13} = L_{32} = L_{21} = 1$ with $y > 0$.
    For notation to allow for easier comparisons later on, we will denote these thresholds as $q_1^{132},q_2^{132},q_3^{132}$.
    \begin{align*}
    q_1^{132} &= y, \\
    q_2^{132} &= \begin{cases} \frac{y}{2} &\text{if } y \geq \frac{1}{2}\left(1 - \alpha_x + \sqrt{(1-\alpha_x)^2 + 8(1-\alpha_L)}\right) \\ \min\{q_3^{132} \; , \; \frac{y^3 + 3y^2 + 3y + (1-\alpha_\LL^2)}{2y^2 + (4+3\alpha_x)y + (2 + \alpha_x(3+\alpha_L))}\} &\text{if } y < \frac{1}{2}\left(1 - \alpha_x + \sqrt{(1-\alpha_x)^2 + 8(1-\alpha_L)}\right), \end{cases} \\
    q_3^{132} &= \begin{cases} \min\{q_2^{132} \; , \; \frac{y^3 + 3y^2 + 3y + (1-\alpha_\LL^2)}{3y^2 + (6+\alpha_x)y + (3+\alpha_x(1+2\alpha_L))}\} &\text{if } y \geq \frac{1}{2}\left(1 - \alpha_x + \sqrt{(1-\alpha_x)^2 + 8(1-\alpha_L)}\right) \\ \min\{q_1^{132} \; , \; \frac{y^2 + 2y + (1-\alpha_L)}{3y + 3 + \alpha_x}\} &\text{if } y < \frac{1}{2}\left(1 - \alpha_x + \sqrt{(1-\alpha_x)^2 + 8(1-\alpha_L)}\right). \end{cases} 
    \end{align*}
%    Note that, in the second case, $q_1^{132} > q_3^{132}$ if $y > \frac{1}{4}\left(-(1+\alpha_x) + \sqrt{(1+\alpha_x)^2 + 8(1-\alpha_L)}\right)$.

\item \textbf{Compressed [0]:} Let $L_{ij} = 0$ for all $i \neq j \in \{1,2,3\}$ with $y > 0$.
    For notation to allow for easier comparison later on, we will denote these thresholds as $q_1^0,q_2^0,q_3^0$.
    \begin{align*}
    q_1^0 &= y, \qquad q_2^0 = \frac{y}{2}, \qquad q_3^0 = \frac{y}{3}.
    \end{align*}
\end{enumerate}
As $q_1^* = y$ for any network construction in this setup, we will compare these 4 networks for the defaulting thresholds for banks 2 and 3 only.  
In order to ease the notation for the remainder of this case study we will focus solely on the setting in which all three banks have gross obligations $1+y$ and for which $L_{12} = L_{23} = L_{31} =: \lambda$ and $L_{13} = L_{32} = L_{21} =: \xi$. 
%To simplify notation in this case study, we will define $\lambda := L_{12} = L_{23} = L_{31}$ and $\xi := L_{13} = L_{32} = L_{21}$.  
Figure~\ref{fig:acemoglu-q*} displays the default thresholds, in excess of the fully compressed system, for the 2nd and 3rd bank ($\max\{q_2^*,q_3^*\}$ and $\min\{q_2^*,q_3^*\}$ respectively) with $\alpha_x = \alpha_L = 0.5$.  First, and notably, the default thresholds are lowest in the fully compressed system.  This is further shown in Figure~\ref{fig:acemoglu-default} in which the optimally compressed network with $\Lambda^\#$ is the fully compressed one.  However, for comparison to~\cite{AOT15}, we also want to investigate the rerouting problem.  By investigating $q^*$ we can compare the stability and resilience of financial networks to \emph{systematic} Bernoulli shocks.  As displayed in Figure~\ref{fig:acemoglu-q*}, and as can be verified analytically (for any $\alpha_x,\alpha_L \in [0,1]$), $q_2^{CC} \leq q_2^{123}$ for any $y \geq 0$.  That is for ``small'' shocks the completely connected system is always more stable and resilient than Ring 123.  For ``large'' shocks with small enough obligations to society $y$, we find that $q_3^{CC} \leq q_3^{123}$; for larger obligations to society the opposite ordering is found.  That is, for ``large'' shocks the total obligations to society can alter the stability and resilience ordering between these two networks.  This, generally, coincides with the robust fragility notion from~\cite{AOT15} in which the completely connected system was more robust to small shocks but more fragile to large shocks.  In contrast, the opposite relations hold between the completely connected network and Ring 132; that is, we find that the ring is more stable and resilient for ``small'' shocks ($q_2^{CC} \geq \max\{q_2^{132},q_3^{132}\}$) but the ordering for ``large'' shocks depends on the obligations to society $y$ (for small enough $y$ then $q_3^{CC} \geq \min\{q_2^{132},q_3^{132}\}$, for large enough $y$ then $q_3^{CC} \leq \min\{q_2^{132},q_3^{132}\}$).  As systematic shocks and heterogeneous financial networks are vital to the consideration of financial stability, optimal compression and rerouting take on new significance since, in this case, the robust fragility results of~\cite{AOT15} no longer hold generally. 
%\strike{since the typical heuristics in the literature will not hold generally.}
\begin{figure}
\centering
\begin{subfigure}[t]{0.47\textwidth}
\centering
\includegraphics[width=\textwidth]{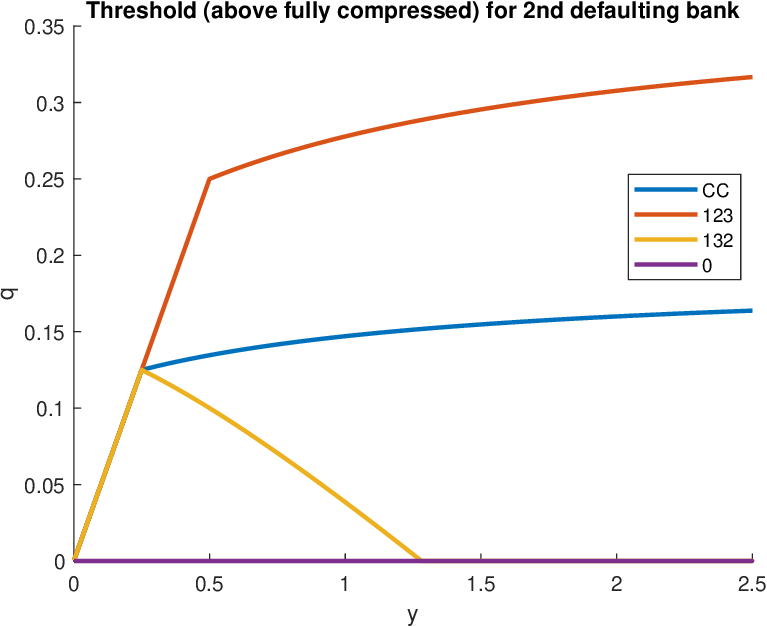}
\caption{The excess (over the fully compressed network) in the default threshold $\max\{q_2^*,q_3^*\} - y/2$ for the second defaulting bank.}
\end{subfigure}
~
\begin{subfigure}[t]{0.47\textwidth}
\centering
\includegraphics[width=\textwidth]{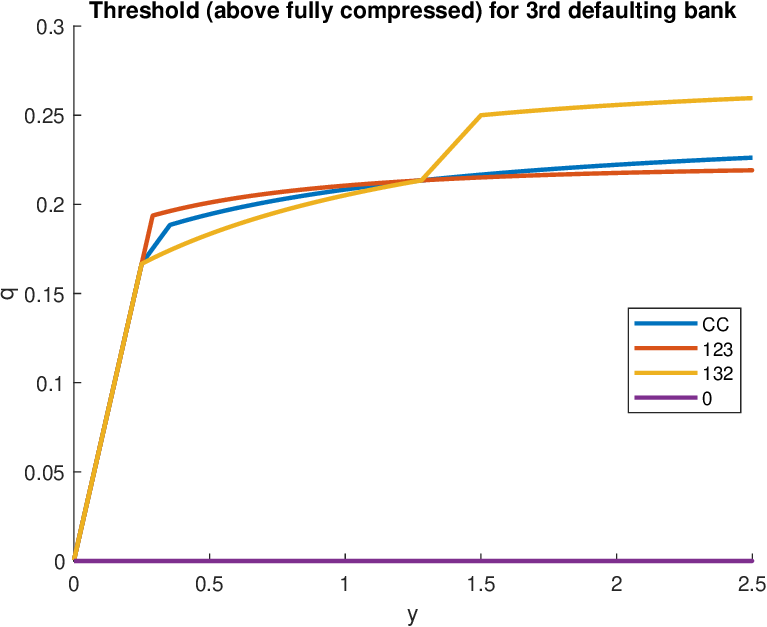}
\caption{The excess (over the fully compressed network) in the default threshold $\min\{q_2^*,q_3^*\} - y/3$ for the third defaulting bank.}
\end{subfigure}
\caption{Section~\ref{sec:acemoglu}: Impact of obligations to society $y$ to the default thresholds $q^*$ for the 4 sample networks: the completed connected system [CC], the two ring networks [123] and [132], and the fully compressed network [0].}
\label{fig:acemoglu-q*}
\end{figure}

\begin{figure}
\centering
\begin{subfigure}[t]{0.47\textwidth}
\centering
\includegraphics[width=\textwidth]{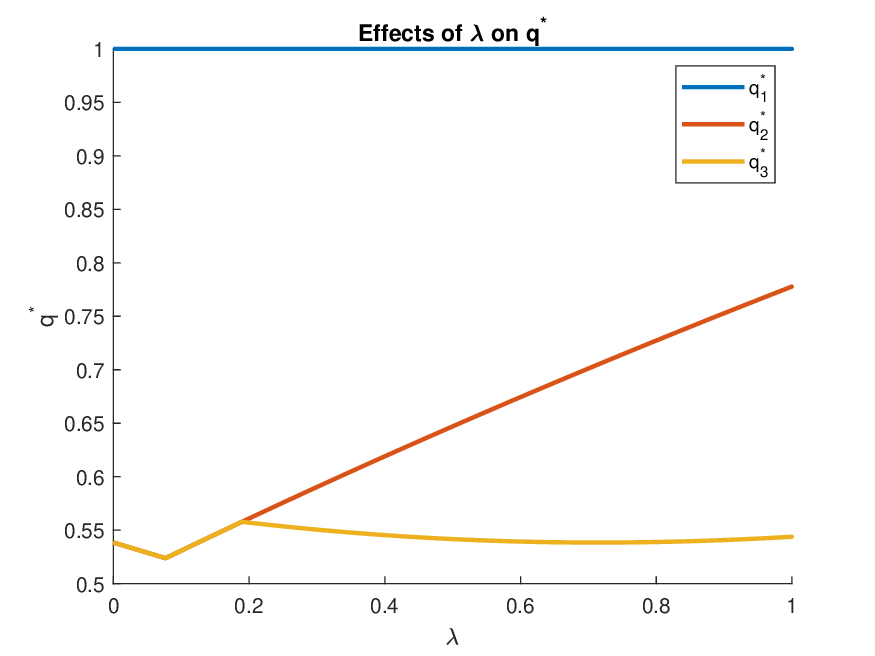}
\caption{Effects of the network topology ($\lambda$) on the default thresholds $q^*$ for rerouting.}
\label{fig:acemoglu-reroute1}
\end{subfigure}
~
\begin{subfigure}[t]{0.47\textwidth}
\centering
\includegraphics[width=\textwidth]{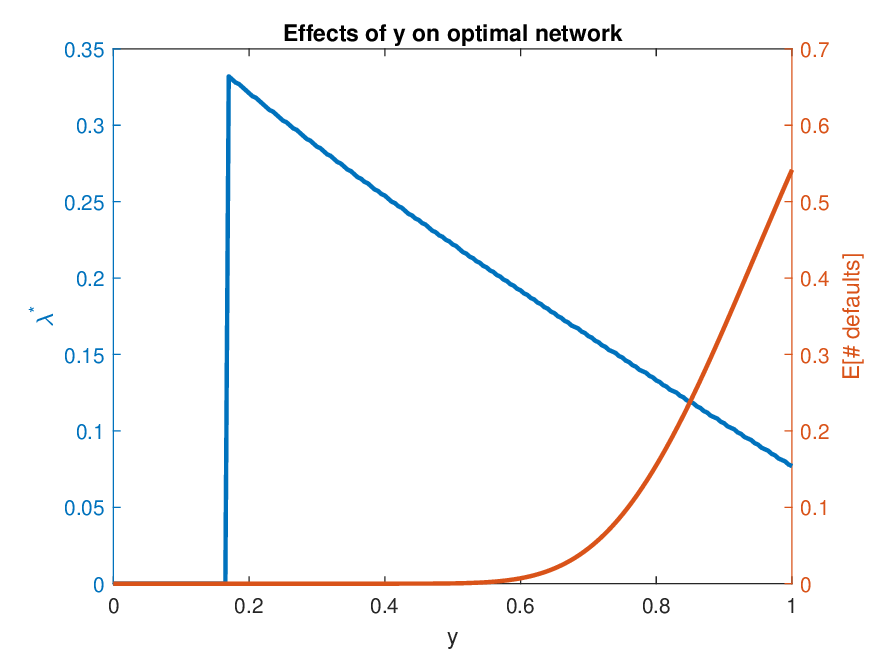}
\caption{Effects of obligations to society ($y$) on the most stable networks (determined by $\lambda^*$) in rerouting.}
\label{fig:acemoglu-reroute2}
\end{subfigure}
\caption{Section~\ref{sec:acemoglu}: Detailed visualization of the optimal rerouting problem.}
\label{fig:acemoglu-reroute}
\end{figure}
We now wish to validate the performance of our genetic algorithm for finding the optimal networks. That is, we wish to demonstrate that the genetic algorithm provides network compressions that equal, or even outperform, those found by other more standard optimization techniques.  We will accomplish this by studying the expectation risk measure $\rho^\E$ (as this was one of the risk measures utilized within~\cite{AOT15}) with all of our sample aggregation functions $\Lambda^\#,\Lambda^\ncal,\Lambda^0$ (recalling that $\Lambda^\scal = \Lambda^\#$ in this uncollateralized setting).  The validation is accomplished by comparing the results of the genetic algorithm with those using an interior point algorithm (initialized at $L_{ij} = \frac{1}{2}$ for all $i \neq j$).\footnote{The interior point algorithm was implemented with ``fmincon'' in MATLAB using the default options.}  We also wish to compare these optimal networks with the 4 sample networks (completely connected, 2 rings, and the fully compressed system) to investigate the optimality of these heuristic constructions.  In Figure~\ref{fig:acemoglu-opt-reroute}, we consider the optimal rerouting problem under change of obligations to society $y$; in Figure~\ref{fig:acemoglu-opt-compress}, we consider the optimal nonconservative compression (with fixed obligations to society) problem under change of obligations to society $y$.  First, and foremost, our genetic algorithm accurately matches or even \emph{outperforms} the optimal network using an interior point algorithm (as seen in optimal nonconservative compression with $\Lambda^0$).  Further, though the heuristic networks coincide with these optimal risk levels in specific cases, they do not uniformly perform as well as the optimal networks.  Most interesting is the consideration of $\Lambda^0$ in which the optimally compressed network nearly coincides with the optimal rerouting problem for low $y$.

\begin{figure}
\begin{subfigure}[t]{0.47\textwidth}
\centering
\includegraphics[width=8cm,height=6cm]{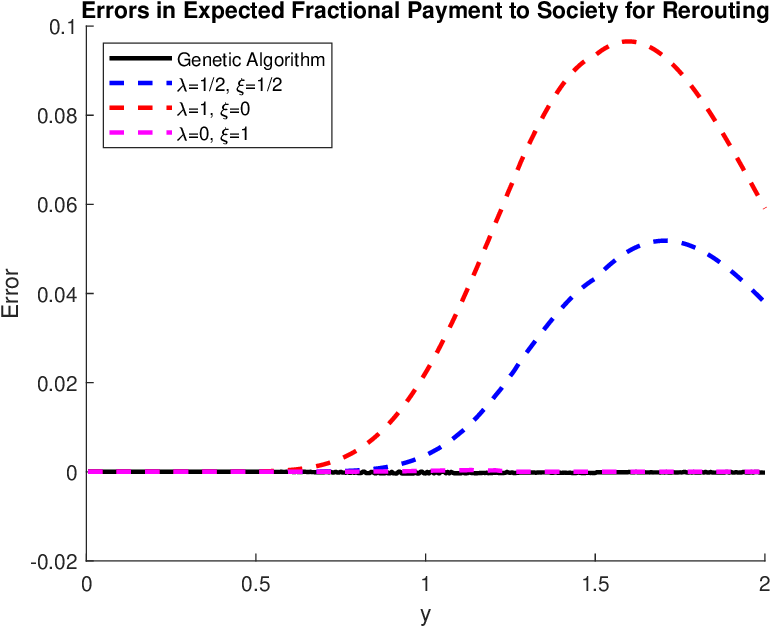}
\caption{Payments to society}
\end{subfigure}
~
\begin{subfigure}[t]{0.47\textwidth}
\centering
\includegraphics[width=8cm,height=6cm]{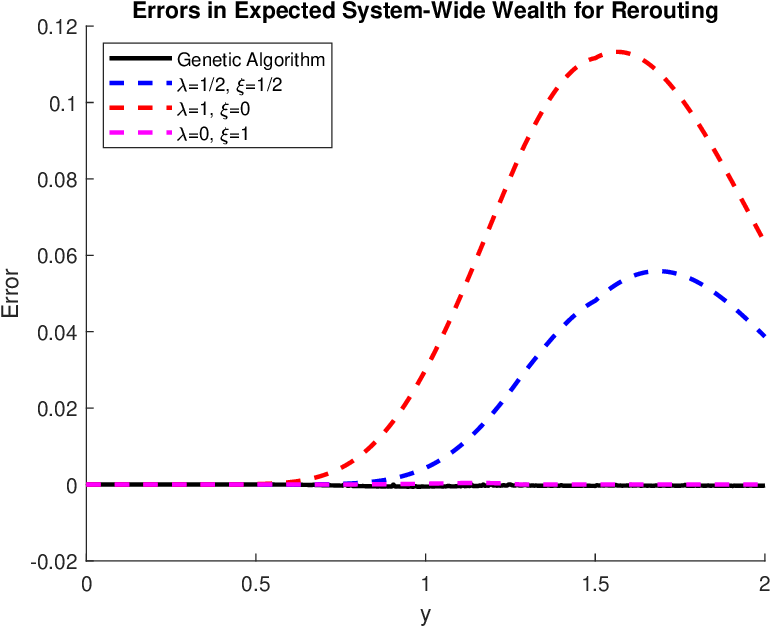}
\caption{Total wealth}
\end{subfigure}

\bigskip
\centering
\begin{subfigure}[t]{0.47\textwidth}
\centering
\includegraphics[width=8cm,height=6cm]{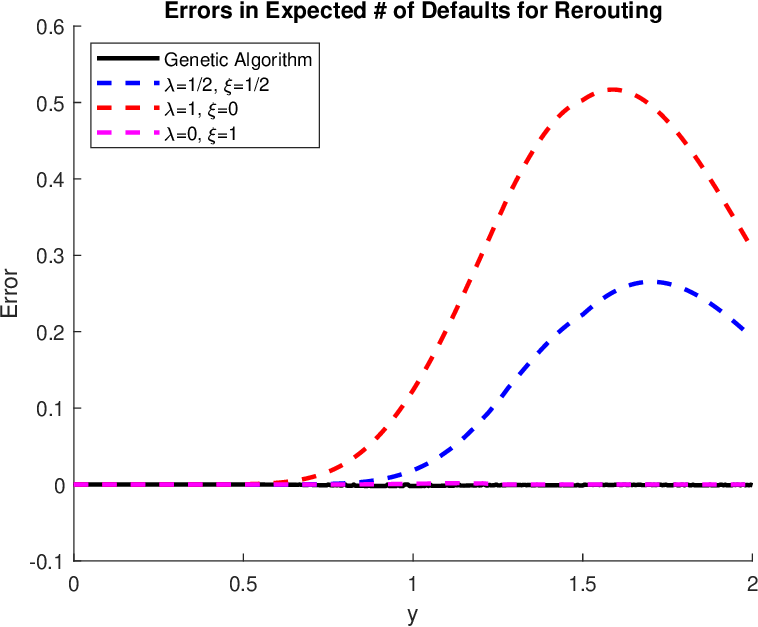}
\caption{Number of defaulting (or solvent) banks}
\end{subfigure}
\caption{Section~\ref{sec:acemoglu}: Validation of the optimal rerouting problem where banks have gross obligations $1+y$ and for which $L_{12} = L_{23} = L_{31} = \lambda$ and $L_{13} = L_{32} = L_{21} = \xi$.}
\label{fig:acemoglu-opt-reroute}
\end{figure}
Consider now the optimal compression and rerouting problems for this three bank system.  As above, throughout this example, we will fix $\alpha_x = \alpha_L = 0.5$ for simplicity of comparison.  First, we will consider the optimal rerouting problem to generalize the notions from~\cite{AOT15}.  In order to ease the notation for the rerouting problem (with gross obligations of $1+y$ for each of the three banks) let $\xi = 1-\lambda$.  First, for the most direct comparison, in Figure~\ref{fig:acemoglu-reroute1} we consider how modifying $\lambda$ affects the defaulting thresholds $q^*$ (with $y = 1$).  By inspection the most stable and resilient system is clearly for some small, but strictly positive, $\lambda$.  These optimally stable networks are described in Figure~\ref{fig:acemoglu-reroute2}; such a system is considered in which $q \sim \mathrm{LogN}(-\frac{\sigma^2}{2},\sigma^2)$ with $\sigma = 20\%$.  Notably, the optimal network depends on the obligations to society $y$.

\begin{figure}
\begin{subfigure}[t]{0.47\textwidth}
\centering
\includegraphics[width=8cm,height=6cm]{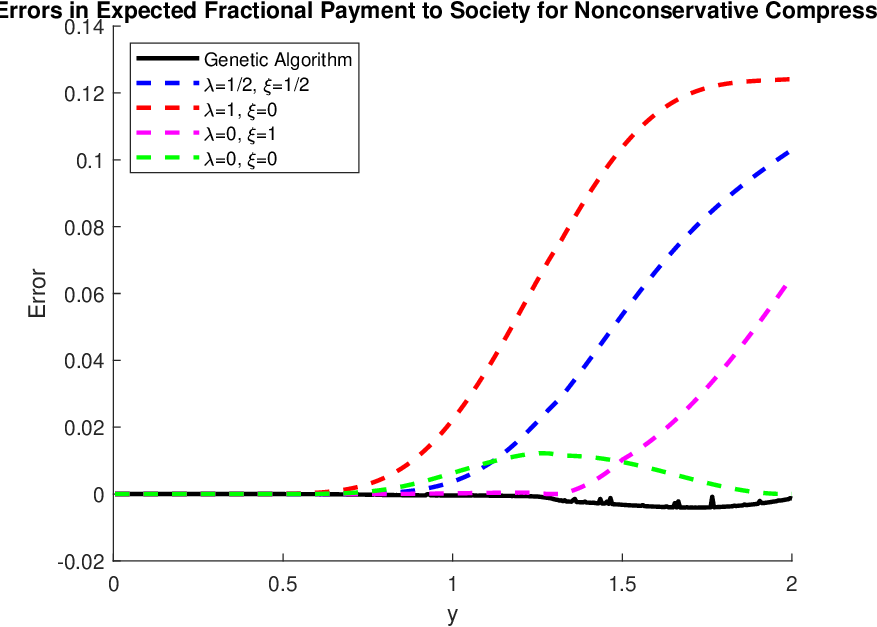}
\caption{Payments to society}
\end{subfigure}
~
\begin{subfigure}[t]{0.47\textwidth}
\centering
\includegraphics[width=8cm,height=6cm]{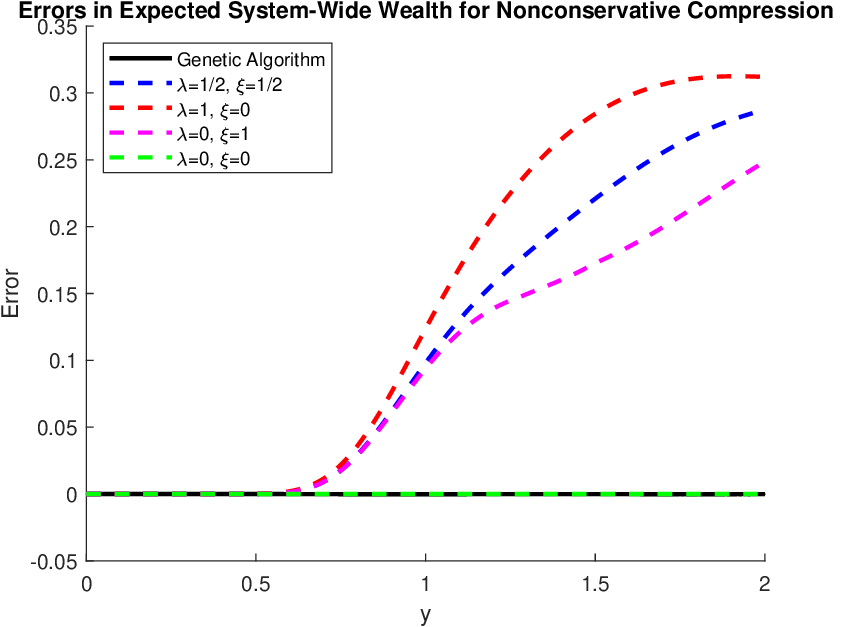}
\caption{Total wealth}
\end{subfigure}
\bigskip

\centering
\begin{subfigure}[t]{0.47\textwidth}
\centering
\includegraphics[width=8cm,height=6cm]{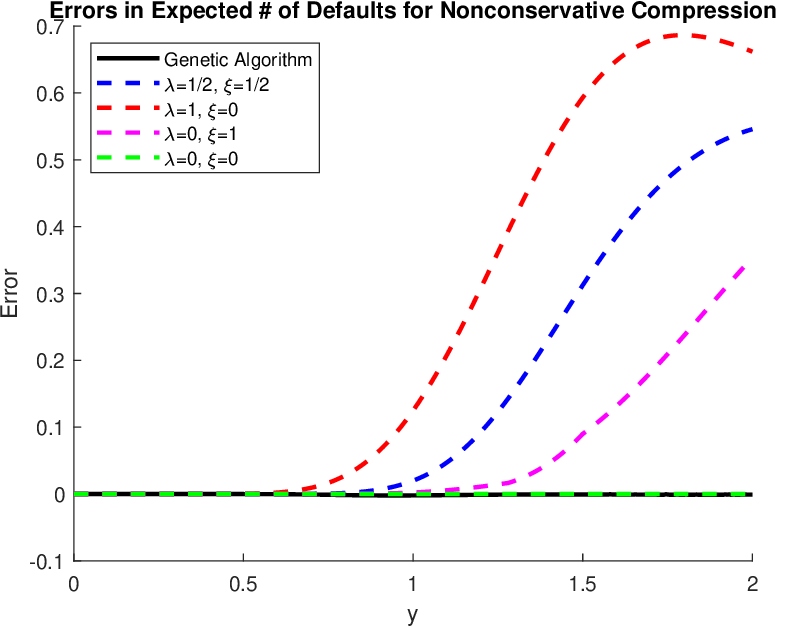}
\caption{Number of defaulting (or solvent) banks}
\label{fig:acemoglu-default}
\end{subfigure}
\caption{Section~\ref{sec:acemoglu}: Validation of the optimal nonconservative compression problem where banks have initial gross obligations $1+y$ and for which $L_{12} = L_{23} = L_{31} = \lambda$ and $L_{13} = L_{32} = L_{21} = \xi$.}
\label{fig:acemoglu-opt-compress}
\end{figure}

\begin{remark}
As presented previously in Remark~\ref{rem:general-stress}, we can consider optimal network compression under general stress scenarios via Monte Carlo simulation.  Intriguingly, and not displayed herein, the optimal network construction and expected risk is generally insensitive to the copula of the stress scenario; this was analyzed numerically over the Gaussian copula with non-negative correlations. 
\end{remark}

\subsection{European banking system}\label{sec:EBA}

In this section we demonstrate how our comonotonic approach with genetic algorithm for optimization can be applied in a larger financial network consisting  of $n = 87$ banks to come from the 2011 European Banking Authority EU-wide stress tests.\footnote{Due to complications with the calibration methodology, we only consider 87 of the 90 institutions. DE029, LU45, and SI058 were not included in this analysis.}  There are several previous empirical studies based on this dataset (see, e.g.,~\cite{CLY14, GV16}) and we calibrate this system by taking the same approach as~\cite{feinstein2017currency}. 
We wish to use this case study for two primary purposes.  First, we use it to demonstrate that our methodology can be applied to large, and realistic, financial systems.  Second, using this calibrated system, we can quantify the suboptimality of the full compression as considered in, e.g.,~\cite{veraart2019does}. Namely, we show that if the systemic risk measures are not explicitly part of the objective function of the optimal compression problem, we end up with suboptimal outcomes or, even worse, an increase in systemic risk. In contrast, we also find that the optimal network compression can find significant systemic risk improvements over the original network.

For the purposes of calibration, we consider a stylized balance sheet for each bank with only three types of assets and liabilities. 
The total \emph{interbank assets} for bank $i$ is $\sum_{j = 1}^n \tilde L_{ji}$ while the total \emph{interbank liabilities} is $\sum_{j = 1}^n \tilde L_{ij}$.  The \emph{external risk-free assets} for bank $i$ is denoted by  $b_i$ and the \emph{external risky assets} are denoted by $s_i$. On the other hand, the \emph{external liabilities} for bank $i$ is $\tilde L_{i0}$ and the bank $i$ is endowed with \emph{capital} $C_i$.

Note that the EBA dataset only provides the total assets $A_i$, capital $C_i$, and interbank liabilities $\sum_{j = 1}^n \tilde L_{ij}$ for each bank $i$. Therefore, we will  make the following simplifying assumptions similar to~\cite{feinstein2017currency,CLY14,GY14}.  
We assume that the external (risky) assets are the difference between the total assets and interbank assets. The external obligations owed to the societal node (denoted by $\tilde L_{i0}$) will be assumed equal to the total liabilities less the interbank liabilities and capital. Further, we assume that the interbank assets are equal to the interbank liabilities for  all banks, i.e., $\sum_{j = 1}^n \tilde L_{ij} = \sum_{j = 1}^n \tilde L_{ji}$ for all $i=1, \dots, n$. 
Under these assumptions,  the remainder of our stylized balance sheet (with collateralization $\mu \in [0,1]$) can be constructed by setting 
\[b_i^{\mu} = 0.8\times(A_i - \sum_{j = 1}^n \tilde L_{ij} - \mu\bar p_i), \quad s_i^{\mu} = 0.2\times(A_i - \sum_{j = 1}^n \tilde L_{ij} - \mu\bar p_i), \]
and,
\[\tilde L_{i0} = A_i - \sum_{j = 1}^n \tilde L_{ij} - C_i, \quad \bar p_i = \tilde L_{i0} + \sum_{j = 1}^n \tilde L_{ij},\]
which will guarantees that firm $i$'s net worth is equal to its capital, i.e., $C_i = A_i - \bar p_i$.  Within this study we consider $\mu \in \{0.0, 0.2, 0.4\}$; we limit $\mu$ in this way so as to guarantee that the assets $b_i^{\mu},s_i^{\mu}$ are positive for every bank $i$.\footnote{The maximum possible collateralization scheme guaranteeing nonnegative assets is $\mu \approx 0.594$.}

We will also need to consider the full nominal liabilities matrix $\tilde L \in \bbr_+^{87 \times 87}$ and not just the total interbank assets and liabilities. To achieve this, we will use the MCMC methodology of~\cite{GV16} which allows for randomized sparse structures to construct the full nominal liabilities matrix consistent with the total interbank assets and liabilities. Under the~\cite{GV16} methodology and the $\sum_{j = 1}^n \tilde L_{ij} = \sum_{j = 1}^n \tilde L_{ji}$ assumption used widely in the literature (see, e.g., \cite{CLY14,GY14}), conservative compression (and therefore also, e.g., nonconservative compression) of the constructed network will result in all interbank obligations being netted out and a zero network remaining due to our initialization of the network. As this example is only for illustrative purposes, we will consider only a single calibration of the interbank network.

The remaining parameters of the system are calibrated as follows. We specify the systematic factor as a lognormal distribution with parameters $q \sim \mathrm{LogN}(r-\frac{\sigma^2}{2},\sigma^2)$ described in millions of euros. Since during the period over which this data was collected, central banks were setting a low interest rate environment, we estimate that the risk-free interest rate is $r = 0$. Finally, from comparisons to annualized historical volatility of European markets in 2011, the volatility of the risky asset is estimated to be $\sigma = 20\%$.

For the purposes of this example, we consider clearing based on the pure collateralized Eisenberg-Noe mechanism, i.e., with full recovery in case of default ($\alpha_x = \alpha_L = 1$).  We will consider two systemic risk measures to optimize over: the 80\% expected shortfall of the payments to society ($\rho^{\rm ES}_{80\%} \circ \Lambda^0$) and the number of banks paying in full ($\rho^{\rm ES}_{80\%} \circ \Lambda^\#$) with stress scenarios $\xcal_{L}^\mu = (b^\mu + 0.8\mu\times[\bar p_i - \sum_{j = 0}^n L_{ij}]) + (s^\mu + 0.2\mu\times[\bar p_i - \sum_{j = 0}^n L_{ij}]) q$ so that any freed up collateral is invested comparably to the original portfolio. We consider expected shortfall herein as it is a coherent risk measure that is widely utilized in practice.  We select these two aggregation functions so as to demonstrate that the conclusions we report below in Tables~\ref{table:EBA-society} and~\ref{table:EBA-solvent} on network compression are not specific to one particular aggregation function but appear more general.  

Herein we compare two types of compression: ``maximally compressed'' corresponds to a compressed network that removes as much excess liabilities from the network as possible (as is considered in~\cite{d2019compressing} and Section~\ref{sec:maximal}) whereas ``optimally compressed'' attempts to minimize the appropriate systemic risk measure.  These two types of compression are then compared over 4 possible compression scenarios: bilateral compression ($\ccal^B(\tilde L)$), conservative compression ($\ccal^C(\tilde L)$), nonconservative compression with fixed obligations to society ($\ccal^N(\tilde L) \cap \ccal^0(\tilde L)$), and nonconservative compression ($\ccal^N(\tilde L)$).  In particular, we wish to compare the systemic risk exhibited by the original network to those found in either the maximally compressed or optimally compressed scenarios.  These results are provided in Tables~\ref{table:EBA-society} and~\ref{table:EBA-solvent}. We wish to remind the reader that the maximally compressed network under conservative, nonconservative-0, and nonconservative compression all result in all interbank obligations being netted out and zero network remaining. Thus, as seen in Tables~\ref{table:EBA-society} and~\ref{table:EBA-solvent}, the gains and losses observed for maximal compression under these three compression algorithms are identical.

As shown in Table~\ref{table:EBA-society}, under maximal compression, the more relaxed the constraints the worse the expected outcome for society in the 20\% tail event; however, by using optimal compression, the systemic risk can be improved significantly under compression.  This is true for all three sampled collateralization levels.  Notably, the optimal bilateral and conservative compression only find marginal improvements in the payments to society, whereas the maximally compressed versions increase systemic risk by billions of euros, in all three collateralization settings.  The nonconservative compressions find substantial benefits, but massive losses when maximally applied without consideration for systemic effects; in the uncollateralized $\mu = 0$ setting, the optimal compression saves over \euro1.6 and \euro5.3 billion, but adds nearly \euro6.75 billion euros of cost when maximally applied.  These uncollateralized networks are visualized in Figure~\ref{fig:networks}; we wish to highlight that optimal bilateral and conservative compression appear quite similar and thus only conservative compression is displayed.  We do not display the maximally compressed networks as, due to the assumptions for the network calibration, the interbank network is zeroed out for conservative, nonconservative-0, and the nonconservative compression.  Finally, we wish to note that though the benefits of compression appear to drop as collateralization increases; this is more than offset by the improvement in systemic risk from the collateralization itself.
\begin{table}
\resizebox{\textwidth}{!}{
\begin{tabular}{|r|r||c|c|c|c|}
\cline{3-6}
\multicolumn{2}{c|}{~} & \textbf{Bilateral} & \textbf{Conservative} & \textbf{Nonconservative-0} & \textbf{Nonconservative}\\ \cline{3-6}\hline
\multirow{2}{*}{\boldmath$\mu=0.0$} & \textbf{Maximally Compressed:} & 2717.449 & 6749.574 & 6749.574 & 6749.574 \\ \cline{2-6}
& \textbf{Optimally Compressed:} & -0.033 & -0.033 & -1639.660 & -5339.976 \\ \hline \multicolumn{6}{c}{}\\[-1em] \hline
%2.304\times10^7
%
\multirow{2}{*}{\boldmath$\mu=0.2$} & \textbf{Maximally Compressed:} & 8176.074 & 20408.285 & 20408.285 & 20408.285 \\ \cline{2-6}
& \textbf{Optimally Compressed:} & -0.208 & -0.208 & -1031.956 & -5062.909 \\ \hline \multicolumn{6}{c}{}\\[-1em] \hline
%2.322\times10^7
%
\multirow{2}{*}{\boldmath$\mu=0.4$} & \textbf{Maximally Compressed:} & 6421.291 & 16659.836 & 16659.836 & 16659.836 \\ \cline{2-6}
& \textbf{Optimally Compressed:} & -0.026 & -0.026 & -304.968 & -4392.229 \\ \hline %\multicolumn{6}{c}{}\\[-1em] \hline
%2.333\times10^7
%%
%\multirow{2}{*}{\boldmath$\mu=0.6$} & \textbf{Maximally Compressed:} &  &  &  &  \\ \cline{2-6}
%& \textbf{Optimally Compressed:} &  &  &  &  \\ \hline \multicolumn{6}{c}{}\\[-1em] \hline
%%
%\multirow{2}{*}{\boldmath$\mu=0.8$} & \textbf{Maximally Compressed:} &  &  &  &  \\ \cline{2-6}
%& \textbf{Optimally Compressed:} &  &  &  &  \\ \hline 
\end{tabular}
}
\caption{Section~\ref{sec:EBA}: Improvements in $\rho^{\rm ES}_{80\%} \circ \Lambda^0$ from the initial network $\tilde L$ in millions of euros (i.e., negative values indicate cost savings) under different collateralization schemes $\mu$. The expected payments in the 20\% tail under the initial network $\tilde L$ (in millions of euros) are $2.304 \times 10^7$, $2.322 \times 10^7$, and $2.333 \times 10^7$ under $\mu = 0.0$, $0.2$, and $0.4$ respectively against total external obligations of $2.338 \times 10^7$.}
%Total obligations of 2.338\times10^7
\label{table:EBA-society}
\end{table}
\begin{figure}
\centering
\begin{subfigure}[t]{0.45\textwidth}
\centering
\includegraphics[width=\textwidth]{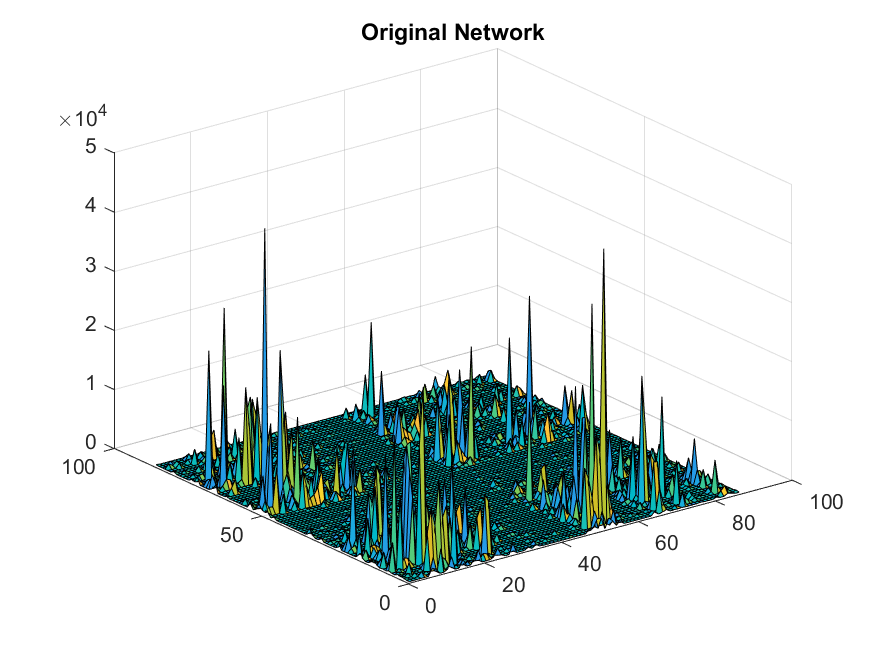}
\caption{Original interbank network}
\end{subfigure}
~
\begin{subfigure}[t]{0.45\textwidth}
\centering
\includegraphics[width=\textwidth]{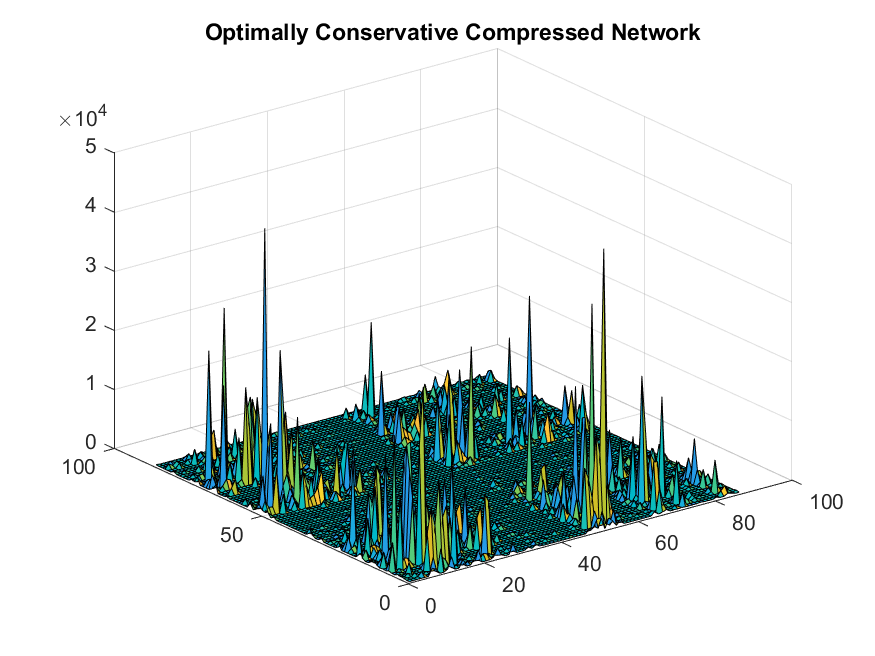}
\caption{Optimal conservative compressed network}
\end{subfigure}
\par\vspace{0.5cm}
\begin{subfigure}[t]{0.45\textwidth}
\centering
\includegraphics[width=\textwidth]{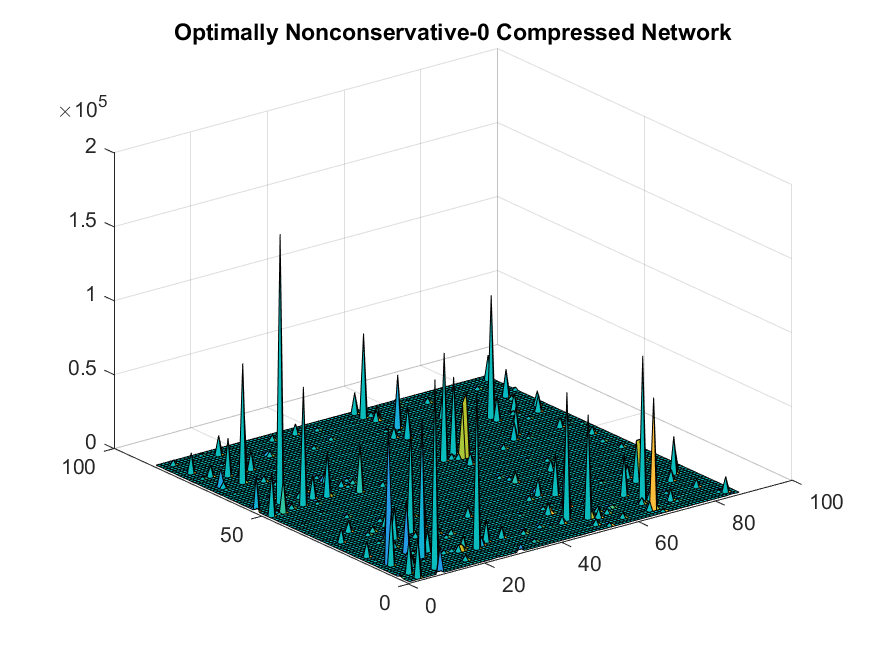}
\caption{Optimal nonconservative-0 compressed network}
\end{subfigure}
~
\begin{subfigure}[t]{0.45\textwidth}
\centering
\includegraphics[width=\textwidth]{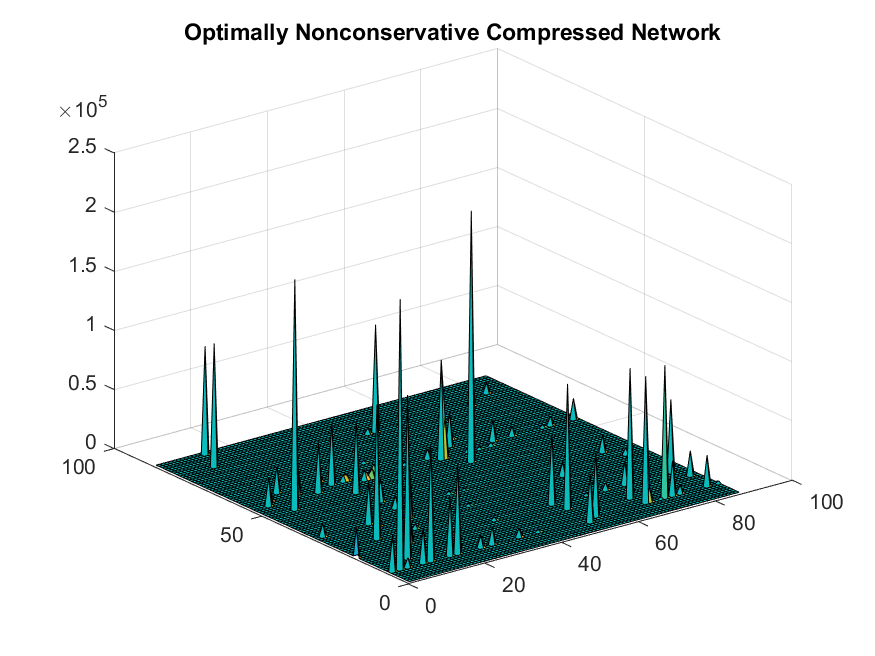}
\caption{Optimal nonconservative compressed network}
\end{subfigure}
\caption{Visualizations of the optimal interbank networks when minimizing $\rho^{\rm ES}_{80\%} \circ \Lambda^0$. Obligations (z-axis) are given in millions of euros. (Banks are ordered by country code.)}
\label{fig:networks}
\end{figure}

In contrast in the uncollateralized setting $\mu = 0$, as shown in Table~\ref{table:EBA-solvent}, the optimal and maximal compression algorithms provide identical systemic risk when attempting to minimize the expected number of defaults in the 20\% tail events.  In this scenario conservative compression outperforms bilateral compression and provides all the benefits of nonconservative compression; in such a setting we therefore find conservative compression the best as it allows all banks to remain with their original intended counterparties.  However, once we include collateralization of obligations, the maximal compression schemes increase the number of defaults whereas optimal compression still reduces the number of defaults in the nonconservative compression scenarios (the improvement in the bilateral and conservative compression settings are marginal in the number of defaults).  Finally, though the benefit of optimal compression appears to shrink (and the harm caused by maximal compression appears to grow) as collateralization increases, this is more than counteracted by the overall drop in the expected number of defaults as collateralization increases.
\begin{table}
\resizebox{\textwidth}{!}{
\begin{tabular}{|r|r||c|c|c|c|}
\cline{3-6}
\multicolumn{2}{c|}{~} & \textbf{Bilateral} & \textbf{Conservative} & \textbf{Nonconservative-0} & \textbf{Nonconservative}\\ \cline{3-6}\hline
\multirow{2}{*}{\boldmath$\mu=0.0$} & \textbf{Maximally Compressed:} & -0.599 & -2.196 & -2.196 & -2.196 \\ \cline{2-6}
& \textbf{Optimally Compressed:} & -0.599 & -2.196 & -2.196 & -2.196 \\ \hline \multicolumn{6}{c}{}\\[-1em] \hline
%46.821
%
\multirow{2}{*}{\boldmath$\mu=0.2$} & \textbf{Maximally Compressed:} & 0.509 & 1.461 & 1.461 & 1.461 \\ \cline{2-6}
& \textbf{Optimally Compressed:} & -0.000 & -0.000 & -0.386 & -0.732 \\ \hline \multicolumn{6}{c}{}\\[-1em] \hline
%27.542
%
\multirow{2}{*}{\boldmath$\mu=0.4$} & \textbf{Maximally Compressed:} & 0.703 & 2.872 & 2.872 & 2.872 \\ \cline{2-6}
& \textbf{Optimally Compressed:} & -0.000 & -0.000 & -0.119 & -0.169 \\ \hline %\multicolumn{6}{c}{}\\[-1em] \hline
%10.349
%%
%\multirow{2}{*}{\boldmath$\mu=0.6$} & \textbf{Maximally Compressed:} &  &  &  &  \\ \cline{2-6}
%& \textbf{Optimally Compressed:} &  &  &  &  \\ \hline \multicolumn{6}{c}{}\\[-1em] \hline
%%
%\multirow{2}{*}{\boldmath$\mu=0.8$} & \textbf{Maximally Compressed:} &  &  &  &  \\ \cline{2-6}
%& \textbf{Optimally Compressed:} &  &  &  &  \\ \hline 
\end{tabular}
}
\caption{Section~\ref{sec:EBA}: Improvements in $\rho^{\rm ES}_{80\%} \circ \Lambda^\#$ from the initial network $\tilde L$ in \# of banks (i.e., negative values indicate fewer banks have unpaid obligations) under different collateralization schemes $\mu$. 
The expected defaults in the 20\% tail under the initial network $\tilde L$ are $46.821$, $27.542$, and $10.349$ under $\mu = 0.0$, $0.2$, and $0.4$ respectively out of $87$ firms.}
\label{table:EBA-solvent}
\end{table}

\begin{remark}
Though we explicitly computed the systemic risk exhibited under maximal compression to demonstrate its (sub)optimality, we could equally have applied Proposition~\ref{prop:derivative} for this purpose.  Notably, the results of Proposition~\ref{prop:derivative} could also be used to verify the (sub)optimality of the maximally compressed network $L^{\max}$.  In fact, the directional derivatives of $\rho^{\rm ES}_{80\%} \circ \Lambda^0$ and $\rho^{\rm ES}_{80\%} \circ \Lambda^\#$ at $L^{\max}$ in the direction of the initial liabilities $\tilde L$ are negative in each case that we find maximal compression is suboptimal and positive when maximal compression is optimal (i.e., $\mu = 0$ under number of banks with unpaid obligations).
\end{remark}

We conclude this discussion by explicitly considering two features of network compression that can harm systemic risk in this case study.  
First, network compression functions like a system with both a senior and junior tranche of debt.  The senior tranche (i.e., obligations that are compressed away) are paid in full leaving all non-payments to occur solely within the junior tranche (i.e., the compressed network).  As we follow the Eisenberg-Noe clearing mechanism with pro-rata repayment, this seniority structure alters the relative liabilities and can spread losses differently in the system; in particular, if the entire interbank network is compressed away then all losses are absorbed by society by construction.  We wish to note that the size of these losses need not be equivalent to those found prior to compression.  
Second, within this case study, we assume that all collateral is held in cash (i.e., not subject to the systematic shock), but held as a risky asset (and thus subject to the systematic shock) when returned to the obligor when a liability is compressed away.  Therefore, as we are focused on the tail of the distribution -- through the use of the expected shortfall risk measure -- the collateral is, in some sense, worth more as collateral than as a risky asset.  This change in the treatment of the collateral assets reduces the total amount of assets available in the system for repayments of obligations.  As such, compression can result in additional losses or defaults.
Both of these effects are seen within the losses reported from maximal compression in Tables~\ref{table:EBA-society} and~\ref{table:EBA-solvent}.

\section{Conclusion}\label{sec:conclusion}
In this work we presented a general formulation for the optimal network compression problem and found it to be NP-hard.  We then focused on an objective function taking the form of systemic risk measures.  In particular, we consider systematic shocks in order to find tractable analytical forms for these systemic risk measures.  Such scenarios allow us to generalize the work of, e.g.,~\cite{AOT15} to consider the robustness of various network topologies. In particular, we found that for heterogeneous networks under systematic shocks, even for a simple heterogeneous three bank system, the robust fragility results of~\cite{AOT15} no longer hold generally. We use a genetic algorithm and show that, in a simple example of three-bank system, the algorithm performs as well as the optimal network found using an interior point algorithm for the rerouting problem and non-conservative compression. Our numerical studies on the European banking system show that if the systemic risk measures are not explicitly part of the objective function of the optimal compression problem, we may end up with suboptimal outcomes or, even worse, an increase in systemic risk. Our illustrative example on this system shows that the optimal network compression can find significant systemic risk improvements over the original network.  
Further studies to systematically investigate the sensitivity of these results to the networks under study would be of general interest. This includes generalizing the robust fragility results of~\cite{AOT15} beyond the symmetric i.i.d.\ assumptions used in that work.

%As the optimal network compression problem is nonconvex and NP-hard in general, the choice of optimization algorithms is of great interest.  This is doubly so if idiosyncratic shocks are introduced as the computation of the systemic risk measures can be NP-hard as well in such a setting (see, e.g.,~\cite{gourieroux2012,BF18comonotonic}).  In this work we implemented a genetic algorithm to search for global minima networks.  We leave further research on choosing optimization procedures, especially machine learning methods, for future research.

%Network compression, as currently implemented by TriOptima and others, seeks to increase operational efficiency and save costs by reducing collateral through netting. As such, the introduction of collateralized networks -- presented in~\cite{ghamami2021collateralized} -- and the changing collateral requirements with the reshaping of the network would be of great importance. We leave this and some other related issues on the effects of network compression on capital requirements and Basel III regulatory constraints for future research.
We wish to highlight a few notable directions to extend the network compression problem.
First, more complex compression rules should be studied; specifically, we may wish to impose different constraints so that compression can only occur over some given subsets of nodes in the initial network, such as compressing liabilities by country. It would be critical to understand the relation between systemic risk reduction and the choice of this subset of the initial financial network that are permitted to compress. 
Another area of interest would be to determine how the presence of multiple CCPs could affect the optimal network compression outcome. Further research may determine the sensitivity of the outcome to the CCPs capitalizations. 
Moreover, designing blockchains and smart contract technology that could integrate network compression remains an open problem. 
Furthermore, as suggested by an anonymous referee, the optimal compression problem can be utilized in a two-stage optimization setting in order to select between different solutions of the maximal compression problem (as, e.g., non-conservative compression generally leads to multiple solutions).
We leave these and some other related issues on the effects of network compression on capital requirements and Basel III regulatory constraints for future research.

\bibliographystyle{apa}
\bibliography{bibtex2}

\appendix
%\section{Proofs}
\section{Proof of Theorem~\ref{thm:NP}}\label{sec:proofNP}

We will consider the minimum relative liability entropy as the objective function:
\begin{equation}
f(L) = -\sum_{i=1}^n \sum_{j=1}^n \pi_{ij} \log (\pi_{ij}),
\end{equation}
where $\pi_{ij}=\frac{L_{ij}}{\bar{p}_i}$ denotes the relative liability of firm $i$ toward firm $j$. We refer to Remark~\ref{rem:entropy} for the interpretation of this objective function for network compression.

Consider an instance of the NP-complete \textsc{subset sum} problem~\cite{karp1972reducibility}, defined by a set of positive integers $S=\{k_1, k_2, \dots, k_n\}$ and an integer target value $\theta \in \N$, we wish to know whether there exist a subset of these integers that sums up to $\theta$. We will show that this can be viewed as a special case for the optimal network compression  in the case of network rerouting, nonconservative and conservative compression models.

Let $K=k_1+k_2+\dots+k_n$ and $\alpha=\theta/K \in (0,1)$ (otherwise, the  \textsc{subset sum}  decision is trivial). Given an instance of the \textsc{subset sum} problem, we define a corresponding instance of the optimal rerouting compression by considering the bipartite network of Figure~\ref{fig:ReroutNP} with two core nodes $\{C_1,C_2\}$ on one side and $n$ periphery nodes $\{P_1, \dots, P_n\}$ on the other side. We set the initial liabilities $\tilde{L}$ as 
$\tilde{L}_{P_i,C_1}=\alpha k_i$  and $\tilde{L}_{P_i,C_2}=(1-\alpha) k_i$ for all $i=1, \dots, n$.  
Note that the total interbank receivables for $C_1$ and $C_2$ is respectively $\theta$ and $K-\theta$, while the total interbank liabilities is zero for $C_1$ and $C_2$. On the other hand, for all $i=1, \dots, n$, the total interbank liabilities for node $P_i$ is $k_i$ while the total interbank receivables is zero.

The optimal rerouting compression model is thus equivalent to finding $x_i\in[0,1]$ for $L_{P_i,C_1}=x_ik_i$ and $L_{P_i,C_2}=(1-x_i)k_i$ which satisfies
$\sum_{i=1}^n x_i k_i =\theta$, and minimizes 
\begin{align*}
f(L)=f(x_1,\dots,x_n)=&-\sum_{i=1}^n x_i \log\left(x_i\right)-\sum_{i=1}^n (1-x_i) \log\left(1-x_i \right).
%\\=&-\sum_{i=1}^n k_i \log\left(k_i\right)-\sum_{i=1}^n \bigl[x_i \log(x_i) + (1-x_i) \log(1-x_i) \bigr] k_i.
\end{align*}
Since $x \log(x) + (1-x) \log(1-x) \leq 0$ for all $x\in[0,1]$ with equality only for $x=0,1$, we infer 
$f(x_1,\dots,x_n) \geq 0$,
and the equality holds if and only if there exists $x_i\in\{0,1\}$ such that $\sum_{i=1}^n x_i k_i =\theta$. Hence, if the solution to optimal rerouting compression model corresponds to $f(L)=0$ then  there exist a subset of $S$ that sums up to $\theta$.

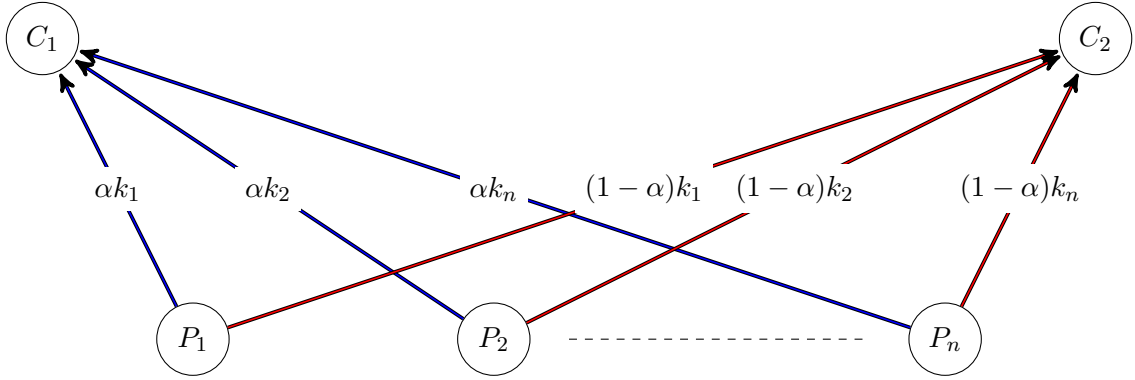
\begin{figure}
\centering
\begin{tikzpicture}[>=stealth',shorten >=1pt,node distance=4cm,on grid,initial/.style    ={}]

                  \node[state] (v1) at (0,0) {$P_1$}; 
        \node[state] (v2) at (4,0)  {$P_2$};
        %\node[state] (vn1) at (8,0) {$v_{n-1}$};
         \node[state] (vn) at (10,0) {$P_{n}$};
         \node[state] (c1) at (-2,4) {$C_1$};
         \node[state] (c2) at (12,4) {$C_2$};

        \draw[dashed] (5, 0) -- (9, 0);
        
        \tikzset{mystyle/.style={->,double=blue}} 
\tikzset{every node/.style={fill=white}}

        \path[->] (v1) edge [mystyle]   node {$\alpha k_1$} (c1);
       \path[->] (v2) edge [mystyle]   node {$\alpha k_2$} (c1);
       \path[->] (vn) edge [mystyle]   node {$\alpha k_n$} (c1); 
       
            \tikzset{mystyle/.style={->,double=red}} 
\tikzset{every node/.style={fill=white}} 
       
       \path[->] (v1) edge [mystyle]   node {$(1-\alpha) k_1$} (c2);
       \path[->] (v2) edge [mystyle]   node {$(1-\alpha) k_2$} (c2);
       \path[->] (vn) edge [mystyle]   node {$(1-\alpha) k_n$} (c2); 
        
\end{tikzpicture}

\caption{Reduction to \textsc{subset sum} NP-complete problem for optimal rerouting and nonconservative compression.}
\label{fig:ReroutNP}
\end{figure}

Further, note that in the bipartite network of Figure~\ref{fig:ReroutNP}, the optimal nonconservative compression is equivalent to optimal rerouting. Hence, the same argument shows that the optimal nonconservative compression is NP-hard.

\begin{figure}
\centering
\begin{tikzpicture}[>=stealth',shorten >=1pt,node distance=4cm,on grid,initial/.style    ={}]

                  \node[state] (v1) at (0,0) {$P_1$}; 
        \node[state] (v2) at (4,0)  {$P_2$};
        %\node[state] (vn1) at (8,0) {$v_{n-1}$};
         \node[state] (vn) at (10,0) {$P_{n}$};
          \node[state] (c0) at (5,4) {$C_0$};
         \node[state] (c1) at (-2,-4) {$C_1$};
         \node[state] (c2) at (12,-4) {$C_2$};

        \draw[dashed] (5, 0) -- (9, 0);

        \tikzset{mystyle/.style={->,double=black}} 
\tikzset{every node/.style={fill=white}}

        \path[->] (c0) edge [mystyle]   node {$ k_1$} (v1);
       \path[->] (c0) edge [mystyle]   node {$ k_2$} (v2);
       \path[->] (c0) edge [mystyle]   node {$ k_n$} (vn); 
        
        \tikzset{mystyle/.style={->,double=blue}} 
\tikzset{every node/.style={fill=white}}

        \path[->] (v1) edge [mystyle]   node {$ k_1$} (c1);
       \path[->] (v2) edge [mystyle]   node {$ k_2$} (c1);
       \path[->] (vn) edge [mystyle]   node {$ k_n$} (c1); 
       
            \tikzset{mystyle/.style={->,double=red}} 
\tikzset{every node/.style={fill=white}} 
       
       \path[->] (v1) edge [mystyle]   node {$k_1$} (c2);
       \path[->] (v2) edge [mystyle]   node {$k_2$} (c2);
       \path[->] (vn) edge [mystyle]   node {$k_n$} (c2); 
       
\tikzset{mystyle/.style={->,relative=true,in=120,out=50,double=black}}
   
 \path[->] (c1) edge [mystyle]   node {$K-\theta$} (c0);  
       
\tikzset{mystyle/.style={->,relative=true,in=-120,out=-50,double=black}}
   
 \path[->] (c2) edge [mystyle]   node {$\theta$} (c0);  
        
\end{tikzpicture}

\caption{Reduction to \textsc{subset sum} problem for optimal conservative compression model.}
\label{fig:ConsNP}
\end{figure}
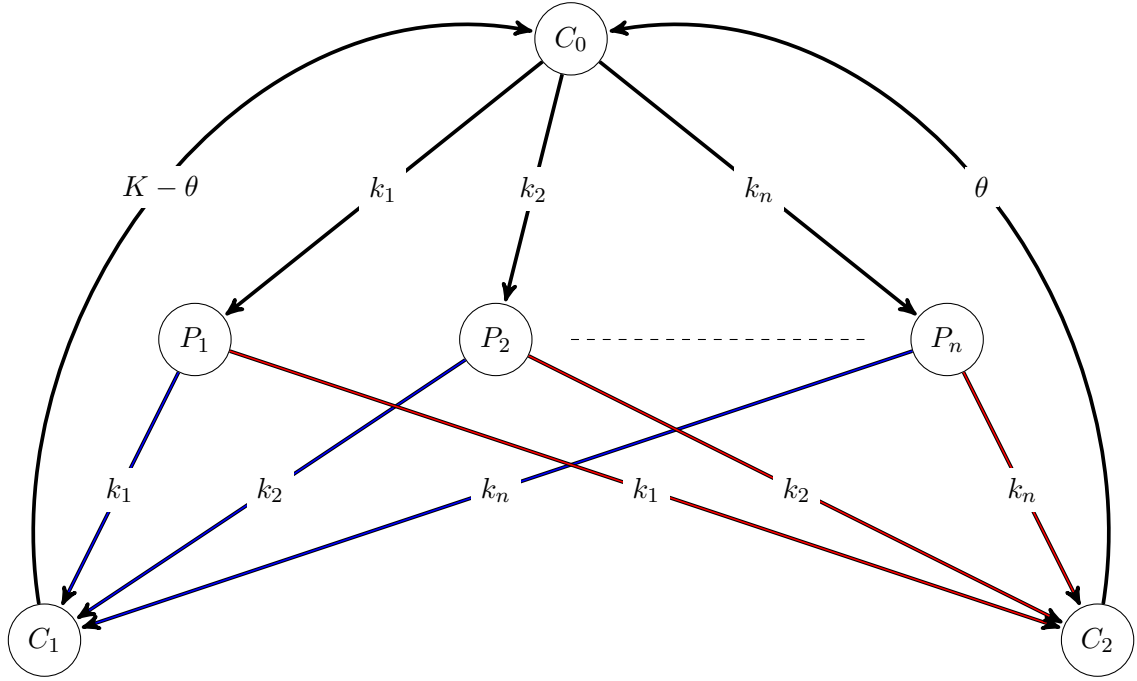

Hence, it only remains to prove the case of conservative compression model.  Given an instance of the \textsc{subset sum} problem, we define a corresponding instance of the  optimal conservative compression by considering the network of Figure~\ref{fig:ConsNP} with  three core nodes $\{C_0, C_1,C_2\}$ and $n$ periphery nodes $\{P_1, \dots, P_n\}$. We set the initial liabilities $\tilde{L}$ as 
$\tilde{L}_{P_i,C_1}=k_i$  and $\tilde{L}_{P_i,C_2}=k_i$ for all $i=1, \dots, n$.  
Note that the net interbank liabilities for $C_1$ and $C_2$ is respectively $-\theta$ and $-(K-\theta)$. On the other hand, for all $i=1, \dots, n$, the net interbank liabilities for node $P_i$ is $k_i$. Further, for the node $C_0$ the net interbank liabilities is zero.

It is then easy to show that the minimum relative liability entropy can be found by setting $L_{C_0, P_i}=0$ for all $i=1, \dots, n$ and consequently, $L_{C_1,C_0}=0$ and $L_{C_2,C_0}=0$. The optimization problem is thus equivalent to find $x_i\in [0,1]$ which gives $L_{P_i,C_1} = x_i k_i$ and $L_{P_i,C_2} = (1-x_i) k_i$ for all $i=1, \dots, n$. The $x_i$s should satisfy $\sum_{i=1}^n x_i k_i = \theta$ and minimizes again
\begin{align*}
f(x_1,\dots,x_n)=&-\sum_{i=1}^n x_i  \log\left(x_i\right)-\sum_{i=1}^n (1-x_i)  \log\left(1-x_i \right) \geq 0.
%\\=&-\sum_{i=1}^n k_i \log\left(k_i\right)-\sum_{i=1}^n \bigl[x_i \log(x_i) + (1-x_i) \log(1-x_i) \bigr] k_i \geq  -\sum_{i=1}^n k_i \log\left(k_i\right).
\end{align*}
The equality holds if and only if there exists $x_i\in\{0,1\}$ such that $\sum_{i=1}^n x_i k_i =\theta$. We conclude that if the solution to optimal conservative compression model corresponds to $f(L)=0$, then  there exist a subset of $S$ that sums up to $\theta$, which completes the proof.

\section{Overview of the genetic algorithm}\label{sec:gp}

The genetic algorithm is an optimization technique to numerically solve nonconvex problems.  Furthermore, this approach does not rely on the differentiability of the objective function.  While this algorithm can still return a local solution, this procedure is less prone to that then, e.g., the gradient descent approach.

Briefly, the idea behind this algorithm is to consider a population of candidate solutions that follow a notion of natural selection. In each step of the algorithm, the fitness (objective value) of all candidate solutions is evaluated; the best performing candidate solutions propagate to the next iteration along with ``children'' and ``mutations'' of these candidate solutions.  This forms a new population of candidate solutions which and the procedure is repeated until some notion of convergence is reached.
Algorithmically, this is presented below for the problem $\min f(x)$:
\begin{enumerate}
\item Generate a (random) population of $m$ feasible solutions $x_1,...,x_m$.
\item\label{step2} Evaluate the fitness of each candidate solution $f(x_1),...,f(x_m)$.
\item If a termination condition is reached (e.g., the best solution found thus far has not changed in $K$ generations) then terminate and return the optimal fitness level and solution.
\item Create the next generation of candidate solutions by:
    \begin{enumerate}
    \item Select $n \leq m$ of the population to consider (e.g., replace the worst performing solutions with copies of the best performing ones).
    \item Recombine some of the new generation of solutions to create new candidates through binary operations.
    \item Randomly alter some of the new generation of solutions through binary operations.
    \end{enumerate}
\item Return to step~\eqref{step2}.
\end{enumerate}

Within the case studies of Section~\ref{sec:casestudy}, we implemented the genetic algorithm using the ``ga'' function in the Global Optimization Toolbox of MATLAB.  Each of these case studies was implemented with the following initial populations of solutions: (i) the uncompressed network $\tilde L$, (ii) 50 random feasible compressed or rerouted networks, and (iii) if a compression problem (i.e., not the rerouting problem), then the maximally compressed network.  The use of the initial and maximally compressed networks in the initial population are to prevent the genetic algorithm from being too biased by the 50 random networks.  The recombining of solutions for the subsequent generation of the algorithm is implemented so as to use a weighted arithmetic mean of the two parent solutions; we utilize the ``crossoverarithmetic'' function in the Global Optimization Toolbox of MATLAB for this purpose.  All other parameters considered are the defaults for ``ga'' in MATLAB.

\section{Systematic shocks}\label{sec:systematic}

While the systemic risk measures provide a meaningful objective to minimize in order to optimize network compression, such constructs present additional computational challenges.  Namely, even a simple systemic risk measure such as $\rho^\E \circ \Lambda^0$ requires an exponential (in number of banks) time to compute explicitly~\cite{gourieroux2012}.  Computationally, this can be overcome with Monte Carlo simulations though that is subject to estimation errors.  Herein we will impose \emph{systematic} shocks on the endowments on the banks, i.e., a comonotonic setting on the stress scenarios $\xcal_L \in (\LL^2)^n$, on an aggregate function based around the collateralized Eisenberg-Noe clearing notion.  This is in contrast to~\cite{AOT15} in which shocks were i.i.d.  

\begin{remark}
For the purposes of this section, we consider a fixed (random) stress scenario, which we denote by $X \in (\LL^2)^n$, that is not dependent on the financial network; we will revisit this assumption to allow for stress scenarios to depend on the network $L \in \lcal$ in the following section.
\end{remark}

\begin{remark}\label{rem:general-stress}
As noted above, any general stress scenario $X \in (\LL^2)^n$ with a law-invariant risk measure can be utilized via Monte Carlo simulation of the systemic risk measure.  As such, though we highlight systematic shocks within this work, the theory developed is applicable more generally.
\end{remark}

Throughout this section let $C: \bbr_+ \to \bbr^n_+$ be a nondecreasing function and $q$ be some random variable such that $C(q) \in (\LL^2)^n$.  The stress scenario is then defined by $X = C(q)$.  For the purposes of this section we will focus on systemic risk measures constructed from Value-at-Risk and expected shortfall (as defined in Example~\ref{ex:riskmsr}) and aggregate functions that depend on the endowments and liability network through Eisenberg-Noe clearing payments only.  We refer to Example~\ref{ex:agg} for a brief discussion of the Eisenberg-Noe clearing problem; importantly, we define the clearing payments $p: \bbr^n_+ \times \lcal \to \bbr^n_{+}$ as a mapping from the endowments and liability network.  As detailed below, this setup allows for polynomial time computation of these meaningful systemic risk measures.  Much of this section follows from the logic of \cite{BF18comonotonic}.

The systematic shock setting allows us to determine threshold market values $q^*$ ($\bar q^*$) such that banks are on the cusp of failing to fulfill their obligations (respectively, such that the banks are on the cusp of defaulting); in particular, we take the view that $q$ denotes a systematic factor.  These values are presented in Definition~\ref{defn:q*} below.  Though presented as a mathematical formulation, \cite[Proposition 4.4]{BF18comonotonic} presents an iterative algorithm for finding $q^*$ (respectively, $\bar q^*$) taking advantage of the monotonicity of $C$.
\begin{definition}\label{defn:q*}
Define $q^*: \lcal \to \bbr^n_+$ so that $q_i^*(L)$ is the minimal value such that firm $i$ is making payments in full %\strike{solvent} 
under the liability network $L \in \lcal$, i.e.\
\[q_i^*(L) = \inf\left\{t \geq 0 \; | \; p_i(C(t),L) \geq \sum_{j = 0}^n L_{ij}\right\}.\]
Similarly, define $\bar q^*: \lcal \to \bbr^n_+$ so that $\bar q_i^*(L)$ is the minimal value such that firm $i$ is not defaulting under the liability network $L \in \lcal$, i.e.\
\[\bar q_i^*(L) = \inf\left\{t \geq 0 \; | \; C_i(t) + \sum_{j = 1}^n \frac{L_{ji}}{\sum_{k = 0}^n L_{jk}} p_j(C(t),L) \geq \sum_{j = 0}^n L_{ij}\right\}.\]
\end{definition}
As noted above, we consider only those aggregate functions $\Lambda$ whose dependence on the endowments $x \in \bbr^n_+$ and liability network $L \in \lcal$ come through the Eisenberg-Noe clearing payments $p(x,L)$, i.e., $\Lambda(x,L) = \bar\Lambda(p(x,L))$ for every $x \in \bbr^n_+$ and $L \in \lcal$ for some monotonic function $\bar\Lambda: \bbr^n \to \bbr$.  In this setting, the threshold values $q^*$ provide a quick heuristic for the health of the financial system.  Notably, if $q_i^*(L) \geq q_i^*(\hat L)$ for every bank $i$ for two financial networks $L,\hat L \in \lcal$, then $\Lambda(C(t),L) \leq \Lambda(C(t),\hat L)$ for any $t \in \bbr_+$ and, thus, $\rho(\Lambda(C(q),L)) \geq \rho(\Lambda(C(q),\hat L))$ for any nonnegative random variable $q$.

Before proceeding to the representations for the systemic risk measures under systematic shocks, we need to introduce some notation that is provided in greater detail in~\cite{BF18comonotonic}.  Namely, we want to consider a piecewise linear construction for the clearing payments which follows from the fictitious default algorithm of~\cite{RV13}.  That is, 
\begin{align*}
%\strike{V(x,L) := \Delta(\ind{V(x,L) < 0},L)x - \delta(\ind{V(x,L) < 0},L)}
p(x,L) := \Delta(\ind{p(x,L) < \sum_{j = 0}^n L_{ij}},L)x - \delta(\ind{p(x,L) < \sum_{j = 0}^n L_{ij}},L)
\end{align*} 
for any endowment $x \in \bbr^n_+$ and liability network $L \in \lcal$.  In this piecewise linear construction, the mappings $\Delta,\delta$ are defined by: 
\begin{align*}
%\strike{\Delta(z,L)} &\strike{:= I + \alpha_x\Pi^\T \left(I - \alpha_L \diag(z)\Pi^\T\right)^{-1}\diag(z)}\\ %\left(I - \left(I - (1-\alpha_L)\diag(z)\right)\Pi^\T \diag(z)\right)^{-1} \left(I - (1-\alpha_x)\diag(z)\right),\\
\Delta(z,L) &:= \alpha_x\left(I - \alpha_L \diag(z)\Pi^\T\right)^{-1}\diag(z), \\
%\strike{\delta(z,L)} &\strike{:= \left[(1-\mu)I - \Pi^\T\left(I-\alpha_L\diag(z)\Pi^\T\right)^{-1}\left(I - (1-\mu)\diag(z)\right)\right]\bar p}\\ %\left(I - \left(I - (1-\alpha_L)\diag(z)\right)\Pi^\T \diag(z)\right)^{-1} \left[I - \left(I - (1-\alpha_L)\diag(z)\right)\Pi^\T\right]\bar p,\\
\delta(z,L) &:= -\left(I - \alpha_L \diag(z)\Pi^\T\right)^{-1}\left(I - (1-\mu)\diag(z)\right)\bar p, \\
\nonumber \bar p &:= L\vec{1}, \qquad \pi_{ij} := \frac{L_{ij}}{\bar p_i},
\end{align*}
for $z \in \{0,1\}^n$ denoting the set of defaulting institutions and $L \in \lcal$ is the liability network. 
In particular, for our comonotonic setting, we can simplify these notions as only a subset of possible sets of defaulting institutions is possible, i.e., we define:
\begin{align*}
\Delta_k(L) &:= \begin{cases} \Delta(\sum_{i = 1}^n \ind{q^*_i(L) \leq q^*_k(L)},L) &\text{if } k = 1,2,...,n \\ 0 &\text{if } k = 0, \end{cases} \\
\delta_k(L) &:= \begin{cases} \delta(\sum_{i = 1}^n \ind{q^*_i(L) \leq q^*_k(L)},L) &\text{if } k = 1,2,...,n \\ -\bar p &\text{if } k = 0. \end{cases}
\end{align*}
Finally, we will use the notation that $[k](L)$ is the index of the $k^{th}$ greatest value of $q^*(L)$, i.e., $q_{[1]}^*(L) \geq q_{[2]}^*(L) \geq ... \geq q_{[n]}^*(L)$.  To simplify following formulae, $q^*_{[0]}(L) \equiv +\infty$ and $q^*_{[n+1]}(L) \equiv 0$ for every liability network $L \in \lcal$.

We are now able to present the specific forms for the systemic risk measures under these systematic shocks.
\begin{proposition}\label{prop:systematic}
Consider a systematic stress scenario described by a nonnegative random variable $q$.
Consider an aggregate function $\Lambda$ whose dependence on the endowments $x \in \bbr^n_+$ and liability network $L \in \lcal$ come through the Eisenberg-Noe clearing payments $p(x,L)$, i.e., $\Lambda(x,L) = \bar\Lambda(p(x,L))$ for every $x \in \bbr^n_+$ and $L \in \lcal$ for some monotonic function $\bar\Lambda: \bbr^n \to \bbr$.
Let $q_{1-\gamma} \in \bbr_+$ denote the $(1-\gamma)$-quantile for $q$, then the Value-at-Risk for level $\gamma \in [0,1]$ can be computed as
\begin{align*}
\rho^{\rm VaR}_\gamma(\Lambda(C(q),L)) &= -\Lambda(C(q_{1-\gamma}),L) = -\bar\Lambda(\Delta_{[k]}(L) C(q_{1-\gamma}) - \delta_{[k]}(L))
\end{align*}
if $q_{1-\gamma} \in [q_{[k+1]}^*(L),q_{[k]}^*(L))$.
Additionally, the expected shortfall for level $\gamma \in [0,1)$ can be computed as
\begin{align*}
\rho^{\rm ES}_\gamma(\Lambda(C(q),L)) &= -\frac{1}{\bbp(q \leq q_{1-\gamma})} \sum_{k = 0}^n \E\left[\Lambda(C(q),L) \ind{q \in [q_{[k+1]}^*(L),q_{[k]}^*(L)) \cap [0,q_{1-\gamma}]}\right]\\
    &= -\frac{1}{\bbp(q \leq q_{1-\gamma})} \sum_{k = 0}^n \E\left[\bar\Lambda(\Delta_{[k]}(L) C(q) - \delta_{[k]}(L))\ind{q \in [q_{[k+1]}^*(L) \wedge q_{1-\gamma},q_{[k]}^*(L) \wedge q_{1-\gamma})}\right].
\end{align*}
\end{proposition}
The proof of this proposition is provided in Appendix~\ref{sec:proofSyst1}.  This result follows directly from the construction of Value-at-Risk and expected shortfall and the logic of Theorem 3.7 of~\cite{BF18comonotonic} with the collateral $\mu$.

We conclude this section with a consideration of a special case in which the expected shortfall can be described in closed form for our example aggregation functions from Example~\ref{ex:agg}.  Specifically, we consider a case in which the systematic factor follows a lognormal distribution.
\begin{corollary}\label{cor:systematic} 
Consider the setting of Proposition~\ref{prop:systematic} in which $q \sim \mathrm{LogN}([r-\sigma^2/2]T,\sigma^2 T)$, $C(t) := b e^{rT} + st$ with $b \in \bbr^n$ and $s \in \bbr^n_+$, and $\Lambda$ takes the form of the specific aggregate functions provided in Example~\ref{ex:agg}.  Let $\Phi$ denote the CDF for the standard normal distribution and $\Phi^{-1}$ is the inverse CDF.  For fixed level $\gamma \in [0,1)$:
\begin{align*}
\rho^{\rm ES}_\gamma(\Lambda^\#(C(q),L)) %&= -\frac{1}{1-\gamma} \sum_{k = 1}^n k \times \P(q \in [q_{[k+1]}^*(L) \wedge q_{1-\gamma} , q_{[k]}^*(L) \wedge q_{1-\gamma}))\\
    &= -\frac{1}{1-\gamma} \sum_{k = 0}^{n-1} (n-k) \times \left(\Phi(-d_{2,[k]}^\gamma(L)) - \Phi(-d_{2,[k+1]}^\gamma(L))\right),\\
\rho^{\rm ES}_\gamma(\Lambda^\scal(C(q),L)) %&= -\frac{1}{1-\gamma} \sum_{k = 1}^n k \times \P(q \in [q_{[k+1]}^*(L) \wedge q_{1-\gamma} , q_{[k]}^*(L) \wedge q_{1-\gamma}))\\
    &= -\frac{1}{1-\gamma} \sum_{k = 0}^{n-1} (n-k) \times \left(\Phi(-\bar d_{2,[k]}^\gamma(L)) - \Phi(-\bar d_{2,[k+1]}^\gamma(L))\right),\\
\rho^{\rm ES}_\gamma(\Lambda^\ncal(C(q),L)) %&= -\frac{1}{1-\gamma} \vec{1}^\T \sum_{k = 0}^n \left[\Delta_{[k]}(L) \E[C(q)\ind{q \in [q_{[k+1]}^*(L) \wedge q_{1-\gamma} , q_{[k]}^*(L) \wedge q_{1-\gamma})}] - \delta_{[k]}(L) \P(q \in [q_{[k+1]}^*(L) \wedge q_{1-\gamma} , q_{[k]}^*(L) \wedge q_{1-\gamma}))\right] \\
    %&\strike{= -\frac{1}{1-\gamma} \vec{1}^\T \sum_{k = 0}^n \Bigl[A^\gamma_k(L)+B^\gamma_k(L)\Bigr],} \\
    &= -\frac{1}{1-\gamma}\sum_{k = 0}^n \vec{1}^\T\left[\left(be^{rT}-\bar p\right)\left(\Phi(-d_{2,[k]}^\gamma(L)) - \Phi(-d_{2,[k+1]}^\gamma(L))\right) \right.\\
   &  \quad \quad \quad \quad \quad  \left.+ s\left(\Phi(-d_{1,[k]}^\gamma(L)) - \Phi(-d_{1,[k+1]}^\gamma(L))\right) + \left(A_k^\gamma(L) + B_k^\gamma(L)\right)\right], \\
\rho^{\rm ES}_\gamma(\Lambda^0(C(q),L)) %&= -\frac{1}{1-\gamma} \sum_{i = 1}^n \pi_{[i]0} \left[\bar p_{[i]} + e_{[i]}^\T\sum_{k = [i]}^n \left[\Delta_{[k]}(L) \E[C(q)\ind{q \in [q_{[k+1]}^*(L) \wedge q_{1-\gamma} , q_{[k]}^*(L) \wedge q_{1-\gamma})}] - \delta_{[k]}(L) \P(q \in [q_{[k+1]}^*(L) \wedge q_{1-\gamma} , q_{[k]}^*(L) \wedge q_{1-\gamma}))\right]\right] \\
    %&\strike{= -\frac{1}{1-\gamma} \sum_{i = 1}^n \frac{L_{[i]0}}{\sum_{j = 0}^n L_{[i]j}} \left[\sum_{j = 0}^n L_{[i]j} + e_{[i]}^\T \sum_{k = [i]}^n \Bigl[A^\gamma_k(L)+B^\gamma_k(L)\Bigr]\right] ,} \\
    &= -\frac{1}{1-\gamma} \sum_{k = 0}^n \sum_{i = 1}^n \frac{L_{i0}}{\sum_{j = 0}^n L_{ij}} e_i^\T \Bigl[A^\gamma_k(L) + B^\gamma_k(L)\Bigr] ,
 \end{align*}
 where $e_j$ is the unit vector with a singular $1$ in its $j^{th}$ component and
\begin{align*}   
   A^\gamma_k(L) &=  \left(\Delta_{[k]}(L) b e^{rT} - \delta_{[k]}(L)\right)\left(\Phi(-d_{2,[k]}^\gamma(L)) - \Phi(-d_{2,[k+1]}^\gamma(L))\right),\\
   B^\gamma_k(L)&= \Delta_{[k]}(L) s \left(\Phi(-d_{1,[k]}^\gamma(L)) - \Phi(-d_{1,[k+1]}^\gamma(L))\right),\\
d_{1,k}^\gamma(L) &= \frac{-\log(q_k^*(L) \wedge q_{1-\gamma}) + (r + \frac{1}{2}\sigma^2)T}{\sigma\sqrt{T}}, \ d_{2,k}^\gamma(L) = d_{1,k}^\gamma(L) - \sigma\sqrt{T}, \\
\bar d_{1,k}^\gamma(L) &= \frac{-\log(\bar q_k^*(L) \wedge q_{1-\gamma}) + (r + \frac{1}{2}\sigma^2)T}{\sigma\sqrt{T}}, \ \bar d_{2,k}^\gamma(L) = \bar d_{1,k}^\gamma(L) - \sigma\sqrt{T}, \\
q_{1-\gamma} &= \exp((r-\sigma^2/2)T + \sigma\sqrt{T}\Phi^{-1}(1-\gamma)).
\end{align*}
%\rd{WHERE $\bar q^*$ IS DEFINED BY $\bar q_i^*(L) := \inf\left\{t \geq 0 \; | \; C_i(t) + \sum_{j = 1}^n \frac{L_{ji}}{\sum_{k = 0}^n L_{jk}} p_j(C(t),L) \geq \sum_{j = 0}^n L_{ij}\right\}$}
\end{corollary}
The proof of this corollary is provided in Appendix~\ref{sec:proofSyst2}.  The form for these systemic risk measures is due to the construction of the expected shortfall as provided in Proposition~\ref{prop:systematic} and the manipulation of the lognormal distribution as in the Black-Scholes pricing formula for European options.

\subsection{Proof of Proposition~\ref{prop:systematic}}\label{sec:proofSyst1}

First we will consider Value-at-Risk with parameter $\gamma \in [0,1]$.  By construction, the comonotonic copula $C$ is nondecreasing in $q$ and, by definition, the aggregation function $\Lambda$ is nondecreasing in $C(q)$.  Therefore, for any $L \in \lcal$ and $\epsilon > 0$, 
\[\bbp\left(\Lambda(C(q),L) \leq \Lambda(C(q_{1-\gamma}+\epsilon),L)\right) \geq \bbp(q \leq q_{1-\gamma}+\epsilon) > 1-\gamma.\]
By definition of Value-at-Risk, this implies $$\rho_\gamma^{\rm VaR}(\Lambda(C(q),L)) \leq -\Lambda(C(q_{1-\gamma}),L).$$  
Additionally, for any $L \in \lcal$ and $\epsilon > 0$,
\[\bbp\left(\Lambda(C(q),L) \leq \Lambda(C(q_{1-\gamma}),L) - \epsilon\right) \leq \bbp(q < q_{1-\gamma}) \leq 1-\gamma.\]
Thus the representation for the Value-at-Risk is proven with the secondary representation following directly by the construction of $\Delta_{[k]}(L),\delta_{[k]}(L)$.

Consider now the expected shortfall with parameter $\gamma \in [0,1]$.  Similar to the Value-at-Risk above
\[\rho_\gamma^{\rm ES}(\Lambda(C(q),L)) = -\E[\Lambda(C(q),L) | q \leq q_{1-\gamma}]\]
for any $L \in \lcal$.  Partitioning $q$-space by $q^*$, we find
\begin{align*}
\rho_\gamma^{\rm ES}(\Lambda(C(q),L)) &= -\frac{\E[\Lambda(C(q),L) \ind{q \leq q_{1-\gamma}}]}{\bbp(q \leq q_{1-\gamma})}\\
    &= -\frac{1}{\bbp(q \leq q_{1-\gamma})} \sum_{k = 0}^n \E[\Lambda(C(q),L) \ind{q \in [q^*_{[k+1]}(L) , q^*_{[k]}(L)) \cap [0,q_{1-\gamma}]}]. 
\end{align*}
The second representation follows directly by the construction of $\Delta_{[k]}(L),\delta_{[k]}(L)$.

\subsection{Proof of Corollary~\ref{cor:systematic}}\label{sec:proofSyst2} 
First, by property of the lognormal distribution, $q_{1-\gamma} = \exp((r-\sigma^2/2)T + \sigma\sqrt{T}\Phi^{-1}(1-\gamma))$ and $\bbp(q \leq q_{1-\gamma}) = 1-\gamma$.  Therefore, for any of our 4 aggregation functions, 
\[\rho_\gamma^{\rm ES}(\Lambda(C(q),L)) = -\frac{1}{1-\gamma} \sum_{k = 0}^n \E[\Lambda(C(q),L) \ind{q \in [q^*_{[k+1]}(L) \wedge q_{1-\gamma} , q^*_{[k]}(L) \wedge q_{1-\gamma})}]\]
by Proposition~\ref{prop:systematic}.

Second, by construction of $d_{1,k}^\gamma(L),d_{2,k}^\gamma(L)$ and the logic from the Black-Scholes formula,
\begin{align*}
\Phi(-d_{1,k}^\gamma(L)) &= \E[q \ind{q \leq q_k^* \wedge q_{1-\gamma}}],\\
\Phi(-d_{2,k}^\gamma(L)) &= \bbp(q \leq q_k^* \wedge q_{1-\gamma}).
\end{align*}
Additionally, by construction, 
\begin{align*}
A_k^\gamma(L) + B_k^\gamma(L) &= \E\left[\left(\Delta_{[k]}(L)\left(b e^{rT} + sq\right) - \delta_{[k]}(L)\right) \ind{q \in [q^*_{[k+1]}(L) \wedge q_{1-\gamma} , q^*_{[k]}(L) \wedge q_{1-\gamma})} \right].
\end{align*}

Consider now the four specific aggregation functions.  For simplicity of notation, let $$\ind{k,L} := \ind{q \in [q^*_{[k+1]}(L) \wedge q_{1-\gamma} , q^*_{[k]}(L) \wedge q_{1-\gamma})}.$$

We have
\begin{align*}
%\strike{\rho_\gamma^{\rm ES}(\Lambda^\#(C(q),L))} &\strike{= -\frac{1}{1-\gamma}\sum_{k = 0}^n \sum_{i = 1}^n \bbp(V_i(C(q),L) \geq 0, \; q \in [q^*_{[k+1]}(L) \wedge q_{1-\gamma} , q^*_{[k]}(L) \wedge q_{1-\gamma}))}\\
  %  &\strike{= -\frac{1}{1-\gamma}\sum_{k = 0}^n \sum_{i = 1}^n \bbp(V_i(C(q),L) \leq 0 | q \in [q^*_{[k+1]}(L) \wedge q_{1-\gamma} , q^*_{[k]}(L) \wedge q_{1-\gamma}))} \\
   % &\strike{\qquad\qquad\qquad\qquad \times \bbp(q \in [q^*_{[k+1]}(L) \wedge q_{1-\gamma} , q^*_{[k]}(L) \wedge q_{1-\gamma}))}\\
   % &\strike{= -\frac{1}{1-\gamma}\sum_{k = 0}^n k \bbp(q \in [q^*_{[k+1]}(L) \wedge q_{1-\gamma} , q^*_{[k]}(L) \wedge q_{1-\gamma}))}\\
   % &\strike{= -\frac{1}{1-\gamma}\sum_{k = 1}^n k \times \left(\Phi(-d_{2,[k]}^\gamma(L)) - \Phi(-d_{2,[k+1]}^\gamma(L))\right),}\\
\rho_\gamma^{\rm ES}(\Lambda^\#(C(q),L)) &= -\frac{1}{1-\gamma}\sum_{k = 0}^n \sum_{i = 1}^n \bbp(p_i(C(q),L) \geq \sum_{j = 0}^n L_{ij}, \; q \in [q^*_{[k+1]}(L) \wedge q_{1-\gamma} , q^*_{[k]}(L) \wedge q_{1-\gamma}))\\
    &= -\frac{1}{1-\gamma}\sum_{k = 0}^n \sum_{i = 1}^n \bbp(p_i(C(q),L) \geq \sum_{j = 0}^n L_{ij} | q \in [q^*_{[k+1]}(L) \wedge q_{1-\gamma} , q^*_{[k]}(L) \wedge q_{1-\gamma})) \\
    &\qquad\qquad\qquad\qquad \times \bbp(q \in [q^*_{[k+1]}(L) \wedge q_{1-\gamma} , q^*_{[k]}(L) \wedge q_{1-\gamma}))\\
    &= -\frac{1}{1-\gamma}\sum_{k = 0}^n (n-k) \bbp(q \in [q^*_{[k+1]}(L) \wedge q_{1-\gamma} , q^*_{[k]}(L) \wedge q_{1-\gamma}))\\
    &= -\frac{1}{1-\gamma}\sum_{k = 0}^{n-1} (n-k) \times \left(\Phi(-d_{2,[k]}^\gamma(L)) - \Phi(-d_{2,[k+1]}^\gamma(L))\right),\\
\rho_\gamma^{\rm ES}(\Lambda^\scal(C(q),L)) &= -\frac{1}{1-\gamma}\sum_{k = 0}^n \sum_{i = 1}^n \bbp(C_i(q) + \sum_{j = 1}^n \pi_{ji} p_j(C(q),L) \geq \sum_{j = 0}^n L_{ij}, \; q \in [\bar q^*_{[k+1]}(L) \wedge q_{1-\gamma} , \bar q^*_{[k]}(L) \wedge q_{1-\gamma}))\\
    &= -\frac{1}{1-\gamma}\sum_{k = 0}^n (n-k)\bbp(q \in [\bar q^*_{[k+1]}(L) \wedge q_{1-\gamma} , \bar q^*_{[k]}(L) \wedge q_{1-\gamma}))\\
    &= -\frac{1}{1-\gamma}\sum_{k = 0}^{n-1} (n-k) \times \left(\Phi(-\bar d_{2,[k]}^\gamma(L)) - \Phi(-\bar d_{2,[k+1]}^\gamma(L))\right),\\
%
%\strike{\rho_\gamma^{\rm ES}(\Lambda^\ncal(C(q),L))} &\strike{= -\frac{1}{1-\gamma}\sum_{k = 0}^n \E\left[\vec{1}^\T V(C(q),L) \ind{k,L}\right]} \\
  %  &\strike{= -\frac{1}{1-\gamma}\sum_{k = 0}^n \E\left[\vec{1}^\T \left(\Delta_{[k]}(L)C(q) - \delta_{[k]}(L)\right) \ind{k,L}\right]} \\
  %  &\strike{= -\frac{1}{1-\gamma} \sum_{k = 0}^n \vec{1}^\T \left[A_k^\gamma(L) + B_k^\gamma(L)\right],} \\
\rho_\gamma^{\rm ES}(\Lambda^\ncal(C(q),L)) &= -\frac{1}{1-\gamma}\sum_{k = 0}^n \E\left[\vec{1}^\T \left(C(q) + \Pi^\T p(C(q),L) - \bar p\right) \ind{k,L}\right] \\
    &= -\frac{1}{1-\gamma}\sum_{k = 0}^n \E\left[\vec{1}^\T \left(C(q) + \Pi^\T\left(\Delta_{[k]}(L)C(q) - \delta_{[k]}(L)\right) - \bar p\right) \ind{k,L}\right] \\
    &= -\frac{1}{1-\gamma}\sum_{k = 0}^n \E\left[\left(\vec{1}^\T C(q) + \vec{1}^\T \left(\Delta_{[k]}(L)C(q) - \delta_{[k]}(L)\right) - \vec{1}^\T\bar p\right) \ind{k,L}\right] \\
    &= -\frac{1}{1-\gamma}\sum_{k = 0}^n \vec{1}^\T\left[\left(be^{rT}-\bar p\right)\left(\Phi(-d_{2,[k]}^\gamma(L)) - \Phi(-d_{2,[k+1]}^\gamma(L))\right)\right. \\ 
    &\left. \quad \quad +  s\left(\Phi(-d_{1,[k]}^\gamma(L)) - \Phi(-d_{1,[k+1]}^\gamma(L))\right) + \left(A_k^\gamma(L) + B_k^\gamma(L)\right)\right], \\
%
%\strike{\rho_\gamma^{\rm ES}(\Lambda^0(C(q),L))} &\strike{= -\frac{1}{1-\gamma}\sum_{k = 0}^n \sum_{i = 1}^n \frac{L_{i0}}{\sum_{j = 0}^n L_{ij}} \E\left[\left(\sum_{j = 0}^n L_{ij} - V_i(C(q),L)^-\right)\ind{k,L}\right]}\\
  %  &\strike{= -\frac{1}{1-\gamma}\sum_{k = 0}^n \sum_{i = 1}^n \frac{L_{i0}}{\sum_{j = 0}^n L_{ij}} \E\left[\left(\sum_{j = 0}^n L_{ij} - e_i^\T \left[\Delta_{[k]}(L)C(q) - \delta_{[k]}(L)\right]^-\right)\ind{k,L}\right]}\\
  %  &\strike{= -\frac{1}{1-\gamma} \sum_{i = 1}^n \frac{L_{i0}}{\sum_{j = 0}^n L_{ij}} \left(\sum_{j = 0}^n L_{ij} - \sum_{k = 0}^n e_i^\T\E\left[\left(\Delta_{[k]}(L)C(q) - \delta_{[k]}(L)\right)^- \ind{k,L}\right]\right)} \\
   % &\strike{= -\frac{1}{1-\gamma} \sum_{i = 1}^n \frac{L_{[i]0}}{\sum_{j = 0}^n L_{[i]j}} \left(\sum_{j = 0}^n L_{[i]j} + \sum_{k = [i]}^n e_{[i]}^\T \E\left[\left(\Delta_{[k]}(L)C(q) - \delta_{[k]}(L)\right) \ind{k,L}\right]\right)}\\
   % &\strike{= -\frac{1}{1-\gamma} \sum_{i = 1}^n \frac{L_{[i]0}}{\sum_{j = 0}^n L_{[i]j}} \left[\sum_{j = 0}^n L_{[i]j} + e_{[i]}^\T \sum_{k = [i]}^n \left[A_k^\gamma(L) + B_k^\gamma(L)\right]\right],} \\
\rho_\gamma^{\rm ES}(\Lambda^0(C(q),L)) &= -\frac{1}{1-\gamma}\sum_{k = 0}^n \sum_{i = 1}^n \frac{L_{i0}}{\sum_{j = 0}^n L_{ij}} \E\left[p_i(C(q),L)\ind{k,L}\right]\\
    &= -\frac{1}{1-\gamma}\sum_{k = 0}^n \sum_{i = 1}^n \frac{L_{i0}}{\sum_{j = 0}^n L_{ij}} e_i^\T \E\left[e_i^\T \left(\Delta_{[k]}(L)C(q) - \delta_{[k]}(L)\right)\ind{k,L}\right]\\
    &= -\frac{1}{1-\gamma}\sum_{k = 0}^n \sum_{i = 1}^n \frac{L_{i0}}{\sum_{j = 0}^n L_{ij}} e_i^\T \left[A_k^\gamma(L) + B_k^\gamma(L)\right],
\end{align*}
as desired.

\end{document}